\documentclass[10pt]{article}
\usepackage[realmainfile]{currfile}
\usepackage[dvipsnames]{xcolor}
\usepackage{graphicx}
\usepackage{amscd}
\usepackage{enumitem}
\usepackage{bm}
\usepackage[french,english]{babel}
\usepackage{caption}
\usepackage[utf8]{inputenc}
\usepackage[T1]{fontenc}
\usepackage{rotating,amssymb,amsmath,amsfonts,delarray}
\usepackage{wasysym}
\usepackage{color}
\usepackage{placeins}
\usepackage{float}
\usepackage{amsthm}

\usepackage{amsfonts}
\usepackage{color}
\usepackage{hyperref}
\usepackage{booktabs}

\usepackage{lipsum} 

\usepackage{graphicx}
\usepackage{caption}
\usepackage{subcaption}
\usepackage{float}

\usepackage{tikz-qtree}
\hypersetup{
	colorlinks=true,
	citecolor=blue,
	linkcolor=blue,
	filecolor=magenta,      
	urlcolor=cyan,
}

\makeatletter
\renewcommand\@biblabel[1]{\textbullet}
\makeatother

\definecolor{amethyst}{rgb}{0.6, 0.4, 0.8}

\usepackage{geometry}
\usepackage{txfonts}
\usepackage[bb=dsserif]{mathalpha}
\usepackage{authblk}

\newtheorem{rem}{Remark}
\newtheorem{deff}{Definition}

\usepackage{thmtools}
\newtheorem{ex}{Example}
\newtheorem{theorem}{Theorem}
\newtheorem{cor}{Corollary}

\usepackage{tikz}
\usetikzlibrary{trees}
\usetikzlibrary{shapes}
\usepackage{float}
\usepackage{wrapfig}
\usepackage{xcolor}
\definecolor{amethyst}{rgb}{0.6, 0.4, 0.8}

\definecolor{codegreen}{rgb}{0,0.6,0}
\definecolor{codegray}{rgb}{0.5,0.5,0.5}
\definecolor{codepurple}{rgb}{0.58,0,0.82}
\definecolor{backcolour}{rgb}{0.9,0.9,0.9}

\usepackage{pgfplotstable}

\usepackage[ruled,vlined]{algorithm2e}

\usepackage{multicol}
\usepackage{tikz, pgfplots}
\pgfplotsset{compat=1.17}
\usetikzlibrary[arrows.meta,bending, calc, positioning, patterns]
\usetikzlibrary{positioning}

\definecolor{Rcolor}{RGB}{150,160,190}
\newcommand{\Rx}{\fontsize{10pt}{12pt}\selectfont
	\raisebox{.3em}{\hspace{1.2em}%
		\llap{\resizebox{1.09em}{.5em}{\color{black}$\bigcirc$}}%
		\llap{\resizebox{1.199em}{.55em}{\color{darkgray}$\bigcirc$}}%
		\llap{\resizebox{1.19em}{.52em}{\color{gray!50}$\bigcirc$}}%
		\llap{\resizebox{1.1em}{.5em}{\color{gray}$\bigcirc$}}%
		\llap{\resizebox{1.25em}{.55em}{\color{gray}$\bigcirc$}}%
	}%
	\hspace{-.85em}%
	\textbf{%
		\textcolor{black}{\textsf{R}}%
		\hspace{-.025em}\raisebox{.01em}{\llap{\textcolor{Rcolor}{\textsf{R}}}}%
}}%
\newbox\rbox
\savebox\rbox{\scalebox{0.1}{\Rx}}

\makeatletter
\def\R{\scalebox{\f@size}{\usebox\rbox}\xspace}
\makeatletter

\newcommand{\vecun}[1]{\,{\boldsymbol{1}_{|{#1}\mathrm{dsc}({#1})|}}}

\begin{document}
\sloppy
	\title{Tree-structured Markov random fields with Poisson marginal distributions}
	\author{Benjamin Côté, Hélène Cossette, Etienne Marceau \and 
		\textit{École d'actuariat, Université Laval, Québec, Canada}
	}
        \date{August 24, 2024}
	\maketitle

\begin{abstract}
		A new family of tree-structured Markov random fields for a vector of discrete counting random variables is introduced. According to the characteristics of the family, the marginal distributions of the Markov random fields are all Poisson with the same mean, and are untied from the strength or structure of their built-in dependence. This key feature is uncommon for Markov random fields and most convenient for applications purposes. The specific properties of this new family confer a straightforward sampling procedure and analytic expressions for the joint probability mass function and the joint probability generating function of the vector of counting random variables, thus granting computational methods that scale well to vectors of high dimension. We study the distribution of the sum of random variables constituting a Markov random field from the proposed family, analyze a random variable’s individual contribution to that sum through expected allocations, and establish stochastic orderings to assess a wide understanding of their behavior. 

	\end{abstract} 

\textbf{Keywords} :
 Multivariate Poisson distributions, undirected graphical models, dependence trees,  supermodular order, convex order, binomial thinning operator.

	\section{Introduction}
	
	\label{sect:Introduction}

 Having graphs underlie multivariate distributions performs a translation of their vast range of topologies to a rich variety of dependence schemes; this is the premise of probabilistic graphical models.  
	The graph, through its edges and its vertices, serves as a representation of the dependence relations knitting the random variables to one another. Among others,	\cite{koller2009probabilistic} and \cite{maathuis2018handbook} dive deeply into probabilistic graphical models, with much emphasis on Bayesian networks and Markov random fields (MRFs), also called Markov networks or undirected graphical models. 
	In \cite{besag1974}, a seminal paper on MRFs, the author defines various families through their conditional distributions: the distribution of a random variable is specified given the value taken by its neighbours according to the underlying graph.
	A family that has drawn particular attention to represent a vector of count random variables is Besag's auto-Poisson MRFs, also called Poisson graphical models, where each vertex's conditional distribution is Poisson, the values taken by the neighbouring vertices' random variables influencing its mean parameter. The model found applications, for instance, for vegetation census \cite{augustin2006using} and cancer mortality analysis \cite{ferrandiz1995spatial}. A praised feature of MRFs is that their underlying graph gives an at-glance understanding of the conditional independence relations between its components. They are widely used, notably in statistical modeling (see \cite{cressie2015statistics} and references therein), 
	but their popularity does not match the one of Bayesian networks, with their vast applications in neural networks and other areas.

	It is quite intuitive to define MRFs through their conditional distributions given their conditional independence property; however, doing so prevents marginals to be set, or even tractable. Changing dependence parameters can have large impacts on marginal distributions; this may be seen as a first drawback of MRFs for dependence modeling. For instance, Besag's auto-Poisson MRFs' marginal distributions are generally unknown. One must employ message passing algorithms to derive them, as discussed in \cite{wainwright2008graphical} and Chapter~15 of \cite{theodoridis2015machine}. Moreover, as investigated in \cite{kaiser2000construction}, specifying vertex-wise conditional distributions does not guarantee that a joint distribution can actually rise from it. This is a second drawback of MRFs for dependence modeling, as stated in \cite{cressie2015statistics}. Besag's auto-Poisson is a famous example of that, as it is cursed to solely incorporate negative dependence:
	while positive dependence seems legitimate according to the conditional distributions, the resulting joint distribution would not be valid. To circumvent this, \cite{kaiser1997modeling} and \cite{yang2013poisson} had to resort to variants of the Poisson distribution.

	This paper introduces a family of MRFs with fixed marginal distributions, meaning that varying the dependence within the MRF will not affect the individual distributions of its components. MRFs from the proposed family are therefore not plagued by the first drawback stated above, hence performing an inversion of the paradigm for MRFs as they were traditionally defined. Such an inversion is briefly discussed in Chapter 3.4 of \cite{wainwright2008graphical}; it is also the idea underlying the design of the MRF described in \cite{pickard1980unilateral}.
	Arising from the fixedness of the marginal distributions is an intelligible parameterization: each parameter affects either marginals or dependence.
	We choose Poisson marginals, thus making the proposed family of MRFs stand as a pendant to Besag's auto-Poisson.  Following \cite{inouye2017review}, the proposed MRFs file under marginal Poisson generalizations while Besag's classify as conditional Poisson generalizations. Incidently, the family of vectors of Poisson random variables presented in \cite{kizildemir2017supermodular} also has a graph-encrypted dependence structure, and sorts with marginal Poisson generalizations. However, it is  not a MRF, for it does not satisfy their characterizing Markov property.
	In this paper, we limit ourselves to trees for the underlying structures to MRFs, but pave the way for future works to extend it to graphs in general.

	The stochastic representation describing the proposed family of MRFs with joint distribution having fixed Poisson marginals is given in Theorem~\ref{th:StoDynamics}. It employs the binomial thinning operator, notably used in the Poisson AR(1) process introduced in \cite{mckenzie1985some}, and mentioned in \cite{davis2021count}. One must not confuse the proposed family of MRFs with Conditional Autoregressive (CAR) models by terminological means, as they rather refer to Besag's auto-gaussian models.  
	The stochastic representation has the advantage of providing explicit expressions of the joint probability mass function and of the joint probability generating function of the proposed MRFs. This feature alone shows that the uncertainty towards the existence of a joint distribution has been dispelled, thus lifting the second drawback stated above. Deriving the joint probability generating function also offers comprehension of the distribution of the sum of the MRFs' components and, notably allows to derive the ordinary generating function of expected allocation (OGFEA), introduced in \cite{blier2022generating}. The OGFEA provides an understanding of the part each component plays in the distribution of their sum. Benefits inherited from the stochastic representation further include the positive dependence tying the random variables, more well-suited to a variety of applications than mandatory negative dependence.  Furthermore, the stochastic representation offers a simple and convenient sampling procedure, allowing to easily apply Monte-Carlo evaluation methods; meanwhile, exact computations remain possible and efficient thanks to the explicit joint probability mass function and joint probability generating function. This, along with the rest, propels the applicability of this new family of MRFs for dependence modeling of count random variables in high-dimension settings.

	The paper is constructed as follows. We present the proposed family of multivariate Poisson distributions defined from tree-structured MRFs in Section~\ref{sect:family} through its general stochastic representation. Some of its elemental properties are then examined, such as the joint probability mass function and the joint probability generating function in Section~\ref{subsect:PMF-PGF}, and the covariance between components in Section~\ref{subsect:Covariance}. Stochastic orderings collating the strength of its dependence schemes are studied in Section~\ref{subsect:StochasticOrderingN}. In Section~\ref{sect:Simulation}, we address how the stochastic representation allows efficient sampling methods. Then, in Section~\ref{sect:Sum}, we look at the distribution of the sum of the components constituting the MRF; discussions on stochastic ordering and expected allocation in this context are provided furthermore. Finally, a numerical example, offering a synthesis of the findings, is presented in Section~\ref{sect:NumericalExamples}.

	\section{The proposed family of MRFs}
	\label{sect:family}
	In this section, we propose a family of MRFs with Poisson marginal distributions and dependence structure encrypted on a tree. 
	One of the key contributions in this paper rests on Theorem \ref{th:StoDynamics}, where we present the stochastic dynamics characterizing this new family, using the binomial thinning operator. 
	Before diving into this element around which revolves the subsequent sections, we provide notations and definitions regarding graph theory.
	
	\subsection{Graph notations and definitions}
	
	A graph $\mathcal{G}$ is a set of $d$ vertices $\mathcal{V} = \{1,\ldots,d\}$ coupled to a set of edges $\mathcal{E}\subseteq\mathcal{V}\times\mathcal{V}$, edges linking vertices together.  All graphs considered in this paper are simple, meaning for $u\in\mathcal{V}$, $(u,u)\not\in\mathcal{E}$, and undirected, meaning that, for $u,v \in \mathcal{V}$, if $(u,v)$ is an element of $\mathcal{E}$, then $(v,u)$ is the same element of $\mathcal{E}$. A path from vertex $u$ to vertex $v$, denoted by $\mathrm{path}(u,v)$, is a sequence of successive edges $e \in \mathcal{E}$, the first one starting at vertex $u$ and the last one ending at vertex $v$. In a path, the same edge cannot appear more than once. A vertex is on a path if it is a component of an edge in that path. A tree, denoted by $\mathcal{T}$, is a connected graph in which no path from a vertex to itself exists. A tree thus contains $d-1$ edges. Only trees are considered in this paper.

	The rooted version of a tree, denoted by $\mathcal{T}_r$, is the tree itself, but with one specific vertex $r \in \mathcal{V}$ labelled as the root. For an undirected tree, one can choose any vertex to be the root. The tree being rooted allows to define (below) the descendants, the children, and the parent of a vertex, as they depend on the chosen root. See Section 3.3 of \cite{saoub2021graph} for details on rooted trees.
	The descendants of vertex $v$, denoted by $\mathrm{dsc}(v),\; v \in \mathcal{V}$, is the set of vertices whose path to the root goes through $v$, that is, for a rooted tree $\mathcal{T}_r$, we have
	$\mathrm{dsc}(v) = \left\{u \in \mathcal{V}:  \exists w\in \mathcal{V}, (v,w) \in \mathrm{path}(u,r) \right\}$, $v\in\mathcal{V}$.
	The children of vertex $v$, denoted by $\mathrm{ch}(v),\; v\in \mathcal{V}$, is the set of vertices descendants of $v$ that are also connected to it by an edge, that is, $\mathrm{ch}(v) = \left\{ j \in \mathrm{dsc}(v): (v,j) \in \mathcal{E}\right\}$, $v\in\mathcal{V}$.
	A vertex is called a leaf if it has no children.
	The parent of vertex $v$, denoted by $\mathrm{pa}(v),\; v\in\mathcal{V}\backslash\{r\}$, is the sole vertex connected by an edge to $v$ that is not its children. The root $r$ has no parent.

	Henceforth, we often refer to three typical shapes of trees: the star, 
	the series tree, 
	and the $\chi$-nary tree. 
	The star is a tree in which one vertex is connected to all the others. The series tree, on the other hand, is a tree where each vertex participates in at most two edges. A $\chi$-nary tree is a tree for which there is a rooting such that all vertices have $\chi$ children, except the leafs; its radius is the length of the path between that root and a leaf. Fig. \ref{fig:TypicalTrees} gives an illustration of these three typical tree shapes.

	\begin{figure}[H]
		\centering
		\begin{subfigure}{0.20\textwidth}
			\centering
			\begin{tikzpicture}[xscale = 0.6, yscale = 0.6, thick]
				\foreach \a in {2,3,4,5,6,7,8}
				{
					\draw [xshift = 1cm] (0,0) -- ({(\a-2)*40}:1);
					\filldraw  [fill = white, draw = Maroon, xshift = 1cm] ({(\a-2)*40}:1) circle (0.3) node {\tiny \a};
				}
				\draw [xshift = 1cm] (0,0) -- ({(9-2)*40}:1);
				\filldraw  [fill = white, draw = White, xshift = 1cm] ({(9-2)*40}:1) circle (0.3) node {\tiny ...};
				\draw [xshift = 1cm] (0,0) -- ({(10-2)*40}:1);
				\filldraw  [fill = white, draw = Maroon, xshift = 1cm] ({(10-2)*40}:1) circle (0.3) node {\tiny $d$};
				\filldraw  [ fill = white, draw = Maroon, xshift = 1cm] (0,0) circle (0.3) node {\tiny 1} ;
			\end{tikzpicture}
			\caption{}
		\end{subfigure}
		\begin{subfigure}{0.25\textwidth}
			\centering
			\begin{tikzpicture}[xscale = 0.5, yscale = 0.5, thick]
				\tikzstyle{vertex}=[circle,draw = Maroon, fill=White,minimum size=10pt,inner sep=0pt]
				
				\foreach \name/\x/\y in {1/0/2, 2/1/2, 3/2/2, 4/1/1, 5/2/1, 6/3/1, d/3/0}
				\node[vertex] (G-\name) at (\x,\y) {\tiny $\name$};
				
				\node[circle,draw = White, fill=White,minimum size=10pt,inner sep=0pt] (G-etc) at (2,0) {\tiny ...};
				\foreach \from/\to in {1/2, 2/3, 3/4, 4/5, 5/6, 6/etc, etc/d}
				\draw (G-\from) -- (G-\to);
			\end{tikzpicture}
			\caption{}
		\end{subfigure}
		\begin{subfigure}{0.31\textwidth}
			\centering
			\begin{tikzpicture}[xscale = 1.20, yscale = 0.5, thick]
				\tikzstyle{vertex}=[shape = ellipse,draw = Maroon, fill=White,minimum size=10pt,inner sep=0pt]
				
				\node[vertex] (G-1) at (1,2) {\tiny $1$};
				\node[vertex] (G-2) at (0,1) {\tiny $1.1$};
				\node[vertex] (G-3) at (2,1) {\tiny $1.\chi$};
				\node[vertex] (G-4) at (-0.5,0) {\tiny $1.1.1$};
				\node[vertex] (G-5) at (0.5,0) {\tiny $1.1.\chi$};
				\node[vertex] (G-6) at (1.5,0) {\tiny $1.\chi.1$};
				\node[vertex] (G-7) at (2.5,0) {\tiny $1.\chi.\chi$};
				
				\foreach \from/\to in {1/2, 1/3, 2/4, 2/5, 3/6, 3/7}
				\draw (G-\from) -- (G-\to);
				
				\node[circle,draw = White, fill=White,minimum size=10pt,inner sep=0pt] (G-etc1) at (1,1) {\tiny...};
				\node[circle,draw = White, fill=White,minimum size=10pt,inner sep=0pt] (G-etc2) at (0,0) {\tiny...};
				\node[circle,draw = White, fill=White,minimum size=10pt,inner sep=0pt] (G-etc3) at (1,0) {\tiny...};
				\node[circle,draw = White, fill=White,minimum size=10pt,inner sep=0pt] (G-etc4) at (2,0) {\tiny...};
			\end{tikzpicture}
			\caption{}
		\end{subfigure}
		\caption{Depiction of three typical tree shapes: (a) $d$-vertex star; (b) $d$-vertex series tree; (c) $\chi$-nary tree of radius 2.}
		\label{fig:TypicalTrees}
	\end{figure}

	\subsection{Stochastic dynamics of the proposed family}

	Let us denote by $\circ$ the binomial thinning operator introduced in \cite{steutel1983integer}. For a random variable $X$ taking values in $\mathbb{N}$, the binomial thinning operator is defined in terms of $X$ as follows:
	\begin{equation}
		\alpha\circ X:= \sum_{i=1}^{X} I_i^{(\alpha)},\quad \alpha \in [0,1],
		\label{eq:BinomThinning}
	\end{equation}
	where $\{I^{(\alpha)}_i, \; i\in\mathbb{N}^*\}$ is a sequence of independent Bernoulli random variables taking 1 with probability $\alpha$, with conventions $\sum_{i=1}^0x_i=0$ and $\mathbb{N}^*=\mathbb{N}\backslash\{0\}$. Before its introduction in the context of INARMA models by the author of \cite{mckenzie1985some} and \cite{mckenzie1988some}, the binomial thinning operator had been used to study self-decomposability of discrete random variables, as in \cite{steutel1979discrete} and in \cite{steutel2003infinite}. 
	See \cite{weiss2008thinning} or \cite{scotto2015thinning} for a study of the binomial thinning operator and other thinning operators.

	The following definition of a MRF is formulated according to the one in Chapter 4.2 of \cite{cressie2015statistics}. 
	\begin{deff}[MRF]
		\label{def:PropertyMRF}
		A vector of random variables $\boldsymbol{X} = (X_v,\,v\in\mathcal{V})$ is a MRF defined on a tree $\mathcal{T} = (\mathcal{V},\mathcal{E})$ if it satisfies the local Markov property, that is, for any two of its components, say $X_{u}$ and $X_{w}$, such that $(u,w) \not\in \mathcal{E}$, 
		\begin{equation}
			X_{u}\perp \!\!\!\perp X_{w}\left| \left\{X_{j},\,(u,j)\in \mathcal{E}\right\}\right.,
			\label{eq:localMarkov}
		\end{equation}
		where $\perp\!\!\!\perp$ denotes conditional independence. 
	\end{deff}
	
	On a tree, the local Markov property is equivalent to the global Markov property, by Lemma 1 of \cite{matuvs1992equivalence}; thus, a tree-structured MRF satisfies
	\begin{equation}
		X_{u}\perp \!\!\!\perp X_{w}\left| \left\{X_{j},\, j\in S(u,w)\right\}\right.,\quad u,w\in\mathcal{V},
		\label{eq:globalMarkov}
	\end{equation}
	with $S(u,w)$ being a separator for $u$ and $w$, that is, a set of vertices such that, for each path from $u$ to $w$, there is at least one participating vertex from that set. Only one path exists from $u$ to $w$ on a tree, so any single vertex on $\mathrm{path}(u,w)$ can act as $S(u,w)$. See \cite{lauritzen1996graphical}, Chapter 3, for a discussion on Markov properties for probabilistic graphical models.

	In the following theorem, we introduce a family of MRFs whose stochastic construction relies on the binomial thinning operator defined in (\ref{eq:BinomThinning}).  
	
	\begin{theorem}[Stochastic representation]
		\label{th:StoDynamics}
		Given a tree $\mathcal{T}=(\mathcal{V},\mathcal{E})$ and a chosen root $r\in\mathcal{V}$, let $\mathcal{T}_r$ be its rooted version. Consider a vector of dependence parameters $\boldsymbol{\alpha} = (\alpha_e, e \in \mathcal{E})$, $\boldsymbol{\alpha}\in[0,1]^d$, and define $\boldsymbol{L} = (L_v, \, v \in \mathcal{V})$ as a vector of independent random variables such that $L_v \sim$~Poisson$(\lambda(1-\alpha_{(\mathrm{pa}(v),v)}))$, $\lambda >0$, for $v\in\mathcal{V}$, with the convention $\alpha_{(\mathrm{pa}(r),r)} = 0$ since the root has no parent. Let $\boldsymbol{N} = (N_v, \, v \in \mathcal{V})$ be a vector of count random variables where its components are defined by 
		\begin{equation}
			N_v = 
            \begin{cases}
                L_r, & \text{if } v=r; \\
                 \alpha_{(\mathrm{pa}(v),v)} \circ N_{\mathrm{pa}(v)} +L_v , & \text{if } v \in dsc(r)
            \end{cases}, \quad v\in\mathcal{V}.
			\label{eq:StoDynamics}
		\end{equation}
		Then, the vector $\boldsymbol{N}$ is a tree-structured MRF whose joint distribution has Poisson marginals of parameter $\lambda$. 
	\end{theorem}
	
	\begin{proof}
		We first prove that $\boldsymbol{N}$ is a MRF. Given (\ref{eq:StoDynamics}), for a given vertex $w\in\mathcal{V}$, $N_w$ is solely influenced by its parent, $N_{\mathrm{pa}(w)}$, and by $L_w$, which is specific to $N_w$ and independent of $\{L_v,\, v\in\mathcal{V}\backslash\{w\}\}$.
	The maximum information about $N_w$ is therefore obtained by knowing the value taken by its parent and by those to whom it is a parent -- its children, that is, all the variables whose vertex is connected to $w$ by an edge. This implies conditional independence of $N_w$ and $N_{j}$, $j \in \mathcal{V}\backslash(\{w,\mathrm{pa}(w)\}\cup\mathrm{ch}(w))$, knowing the value taken by its parent and all its children, that is,
	\begin{equation*}
		N_v \perp \!\!\!\perp N_{j} | \left\{N_k,\, {k\in (\{\mathrm{pa}(v)\}\cup\;\mathrm{ch}(v))} \right\},\quad v\in\mathcal{V},\; j\in\mathcal{V}\backslash(\{v,\mathrm{pa}(v)\}\cup\mathrm{ch}(v)).
	\end{equation*}
	This amounts to (\ref{eq:localMarkov}), the local Markov property. Hence, $\boldsymbol{N}$ is a MRF with respect to Definition~\ref{def:PropertyMRF}, and also satisfies the global Markov property given by (\ref{eq:globalMarkov}).  
We now prove by induction that $N_v\sim$~Poisson$(\lambda)$ for all $v\in\mathcal{V}$. Clearly, $N_r\sim$~Poisson$(\lambda)$. Suppose the statement holds true for $N_{\mathrm{pa}(w)}$, $w\in\mathcal{V}\backslash\{r\}$.  We show that $N_w\sim$~Poisson$(\lambda)$. From the representation in (\ref{eq:StoDynamics}), we know $L_w$ is not used in the construction of $N_{\mathrm{pa}(w)}$, since $w$ is not on $\mathrm{path}(r,\mathrm{pa}(w))$. Because $L_w$ is independent of $\{L_v,\, v\in\mathcal{V}\backslash\{w\}\}$ and also of thinning operations performed on them, we have $(\alpha_{(\mathrm{pa}(w),w)}\circ N_{\mathrm{pa}(w)})$ independent of $L_w$. From (\ref{eq:StoDynamics}), it follows
	\begin{equation*}
		\mathcal{P}_{N_w}(t) = \mathcal{P}_{\alpha_{(\mathrm{pa}(w),w)}\circ N_{\mathrm{pa}(w)}}(t) \times\mathcal{P}_{L_w}(t), \quad t \in [-1,1],
	\end{equation*}
	where $\mathcal{P}_X$ denotes the probability generating function of a random variable $X$. By Theorem \hyperref[th:PropertyBinThinOp]{C.\ref{th:PropertyBinThinOp}(d)}, we obtain
	\begin{equation*}
		\mathcal{P}_{N_w}(t) = \mathcal{P}_{N_{\mathrm{pa}(w)}} (1-\alpha_{(\mathrm{pa}(w),w)} + \alpha_{(\mathrm{pa}(w),w)} t) \times \mathcal{P}_{L_w}(t), \quad t \in [-1,1],
	\end{equation*}
	which becomes
	\begin{equation}
		\mathcal{P}_{N_w}(t) = \mathrm{e}^{\lambda( 1-\alpha_{(\mathrm{pa}(w),w)} + \alpha_{(\mathrm{pa}(w),w)} t - 1)} \times  \mathrm{e}^{\lambda(1-\alpha_{(\mathrm{pa}(w),w)})(t-1)}, \quad t \in [-1,1],\label{fgpNj}
	\end{equation}
	from the probability generating functions of $L_w$ and $N_{\mathrm{pa}(w)}$ given the induction hypothesis. It follows from (\ref{fgpNj}) that 
	\begin{equation*}
		\mathcal{P}_{N_w}(t) = \mathrm{e}^{\lambda(t - 1)} \quad  t \in [-1,1],
	\end{equation*}
	which specifies that $N_w$ follows a Poisson distribution of parameter $\lambda>0$.
	Hence, inductively on the parent-child relationships of vertices, we deduce $N_v\sim$~Poisson$(\lambda)$ for all $v\in\mathcal{V}$. 
\end{proof}

The representation in (\ref{eq:StoDynamics}) exhibits that each $N_v$, $v \in \mathcal{V}\backslash\{r\}$, comprises two independent parts: a propagation part and an innovation part. The propagation part, denoted by $\alpha_{(\mathrm{pa}(v),v)}\circ N_{\mathrm{pa}(v)}$, represents the number of events having taken place at the parent vertex of $v$ that have propagated through the edge connecting the two. These events contribute to the number of events $N_v$. The parameter $\alpha_{(\mathrm{pa}(v),v)}$ then represents, for a given event, the probability of such a propagation to occur. We later show that $\alpha_{(\mathrm{pa}(v),v)}$ incidentally corresponds to the Pearson correlation coefficient between $N_{\mathrm{pa}(v)}$ and $N_v$. The innovation part, denoted by $L_v$, represents the number of events originally associated to the vertex $v$. The occurrence of these $L_v$ events is not influenced by the occurrence of events at other vertices in $\mathcal{V}\backslash\{v\}$. The propagation part and the innovation part are independent, given the independence between $L_u$'s, $u\in\mathcal{V}$.
The representation in (\ref{eq:StoDynamics}) is schematized in the following example. 

\begin{figure}[H] 
		\centering
		\begin{tikzpicture}[xscale = 0.5, yscale = 0.5, thick]
			\tikzstyle{vertex}=[circle, draw = Maroon,fill=White,minimum size=10pt,inner sep=0pt]
			
			\foreach \name/\x/\y in {1/0/3, 2/-0.5/2, 3/0.5/2, 4/0/1, 5/1/1, 6/-0.5/0, 7/0.5/0}
			\node[vertex] (G-\name) at (\x,\y) {\tiny \name};
			
			\foreach \from/\to in {1/2, 1/3, 3/4, 3/5, 4/6, 4/7}
			\draw (G-\from) -- (G-\to);
		\end{tikzpicture}
		\caption{The 7-vertex rooted tree $\mathcal{T}_1$ of Example \ref{ex:StoDynamics}.}
		\label{fig:ExTreeNotations}
	\end{figure}

\begin{ex}
	\label{ex:StoDynamics}
	Consider $\boldsymbol{N}=(N_v,\, v\in\{1,\ldots,7\})$, a MRF with joint distribution having fixed Poisson marginals of parameter $\lambda$, on the rooted tree $\mathcal{T}_1$ depicted in Fig.~\ref{fig:ExTreeNotations}, and with dependence parameters $\boldsymbol{\alpha}=(\alpha_e,\,e\in\mathcal{E})$, as in Theorem~\ref{th:StoDynamics}. The stochastic representation of $\boldsymbol{N}$ can be illustrated as follows:
	\begin{equation*}
		N_1 = L_1 
		\Rightarrow
		\begin{cases}
			N_2 = \alpha_{(1,2)} \circ N_1 + L_2, \\
			N_3 = \alpha_{(1,3)} \circ N_1 + L_3, 
			\Rightarrow
			\begin{cases}
				N_4 = \alpha_{(3,4)} \circ N_3 + L_4, 	
				\Rightarrow
				\begin{cases}
					N_6 = \alpha_{(4,6)} \circ N_4 + L_6, \\
					N_7 = \alpha_{(4,7)} \circ N_4 + L_7,
				\end{cases} \\
				N_5 = \alpha_{(3,5)} \circ N_3 + L_5, 
			\end{cases}
		\end{cases}
	\end{equation*}
	with independent random variables $L_{1},\ldots, L_7$, such that $L_{1}\sim$~Poisson$(\lambda)$ and $L_i\sim$~Poisson$(\lambda(1-\alpha_{(\mathrm{pa}(i),i)}))$ for $i\in\{2,\ldots,7\}$. 
\end{ex}

One could be tempted to express the representation given in (\ref{eq:StoDynamics}) entirely in terms of $\boldsymbol{L}$. For instance, $N_4$ in Example~\ref{ex:StoDynamics} becomes $N_4=\alpha_{(3,4)}\circ(\alpha_{(1,3)}\circ L_1 + L_3) + L_4$ which, by Theorem~\hyperref[th:PropertyBinThinOp]{\ref{th:PropertyBinThinOp}(f)}, yields $N_4\stackrel{d}{=}(\alpha_{(3,4)}\alpha_{(1,3)})\circ L_1 + \alpha_{(3,4)}\circ L_3 + L_4$. Note that the property given by Theorem~\hyperref[th:PropertyBinThinOp]{\ref{th:PropertyBinThinOp}(f)} holds only in terms of distribution. Therefore, representing $N_4$ this way, one would forget that the thinning operation performed on $L_1$ must have the same realisation for $N_3$ as for $N_4$, according to (\ref{eq:StoDynamics}); thus, one would be overlooking a portion of the dependence relation tying $N_3$ and $N_4$, or $N_v$ and $N_{\mathrm{pa}(v)}$ in general, $v\in\mathcal{V}$, and should therefore not take this shortcut to derive results on the joint distribution.

As explained in Chapter 15 of \cite{theodoridis2015machine}, determining the marginal distributions of a tree-structured MRF usually requires the use of message passing algorithms, as one typically cannot directly extract this information from its stochastic construction or the parameterization of its probability mass function. Message passing algorithms yield exact results on trees (see for instance \cite{wainwright2003tree}). It is nonetheless a large step in efficiency to have fixed marginal distributions, as we achieve with the construction in Theorem~\ref{th:StoDynamics}, eluding the need to resort to any algorithm to derive them.

One may notice similarities between the proposed family of MRF and the Poisson AR(1) stochastic process. Poisson AR(1) processes were introduced in \cite{mckenzie1985some} and greatly studied in the context of time series for count data, starting with \cite{mckenzie1988some}. The stochastic representation of such processes is given by 
\begin{equation}
	N_{i} = \alpha\circ N_{i-1} + L_i, \quad \alpha \in [0,1],\; i \in \mathbb{N}^*\backslash\{1\}, 
	\label{eq:PoissonAR1}
\end{equation}
where $N_1 \sim$~Poisson$(\lambda)$, and with the sequence of independent random variables $\{L_i,\, i\in\mathbb{N}^{*}\backslash\{1\}\}$, independent of $N_1$ and such that $L_i\sim$~Poisson$(\lambda(1-\alpha))$, $\lambda >0$. 
Juxtaposing (\ref{eq:StoDynamics}) and (\ref{eq:PoissonAR1}), one observes that the stochastic construction in Theorem \ref{th:StoDynamics} on a series tree is equivalent to a Poisson AR(1) process over a finite time set $\{1,...,d\}$. The Poisson AR(1) process hence corresponds to the limit case, when the number of vertices becomes large, of a MRF constructed according to Theorem \ref{th:StoDynamics} on a series tree.

The family of multivariate distributions derived from the construction procedure of Theorem \ref{th:StoDynamics} encapsulates, within the edges of the tree, the propagation dynamics of the Poisson AR(1) process, leaning on the (discrete) self-decomposability of the Poisson distribution as expressed in \cite{steutel2003infinite}, Chapter 5.4. The fixed Poisson marginals arise from this property. This rendition of the propagation dynamics as a MRF widens extensibly the  dependence structure possibilities within a vector of Poisson random variables given the numerous feasible tree constructions. This allows to grow out of the limitation to temporal-like dependence by being able to express more general dependence schemes, spatial dependence settings for example (\cite{cressie2015statistics}, Chapter~4). The possibility of having different dependence parameters $\alpha_e$ for each edge $e \in \mathcal{E}$ adds another layer of flexibility.

We must stress the distinction between the family of MRFs constructed as in Theorem~\ref{th:StoDynamics} and auto-Poisson MRFs, introduced in \cite{besag1974} and also called Poisson Graphical Models. Auto-Poisson MRFs do not have Poisson marginal distributions; it is their conditional distributions that are Poisson. The authors of \cite{inouye2017review} explain the clear distinction between multivariate distributions with Poisson marginals and with conditional Poisson distributions. The proposed family of MRFs is also distinct of Poisson dependency networks, introduced in \cite{hadiji2015poisson}, which are a non-parametric analogue to Besag's auto-Poisson.

\begin{theorem}[Choice of the root]
	\label{th:Reversibility}
	Let $\boldsymbol{N} = (N_v, \, v \in \mathcal{V})$ be a tree-structured MRF with joint distribution having fixed Poisson marginals as in Theorem \ref{th:StoDynamics}. The choice of the root of $\mathcal{T}$ has no stochastic incidence. Two rooted versions of the same tree yield the same joint distribution of $\boldsymbol{N}$.   
\end{theorem}
\begin{proof}
	By selecting root $r^{\prime}$ rather than root $r$, for $r,r^{\prime}\in\mathcal{V}$, the parent-child relationship switches on $\mathrm{path}(r,r^{\prime})$. Children vertices not on that path all retain their original parent vertex, hence their original stochastic dynamics, from the global Markov property established in Theorem~\ref{th:StoDynamics}.
	Thus, the reversibility of the stochastic dynamic on a series tree is all that is required to state that the choice of the root has no incidence on the joint distribution. The time-reversibility of the Poisson AR(1) process, as proved in \cite{mckenzie1988some}, ensures that reversibility on a spine. 
\end{proof}

Theorem~\ref{th:Reversibility} shows that the directionality seemingly implied by the binomial thinning operator in the construction given in (\ref{eq:StoDynamics}) is illusory. The parent-child dynamics only offer convenience for notation, proofs and programming purposes, but have no impact on a strictly stochastic basis. 

Let us denote by $\mathbb{MPMRF}$ the family of multivariate Poisson distributions defined from tree-structured MRFs, as described by the stochastic representation in Theorem~\ref{th:StoDynamics}. We denote by $\text{MPMRF}(\lambda,\boldsymbol{\alpha}, \mathcal{T})$ a member of $\mathbb{MPMRF}$ having parameter $\lambda$ for its fixed Poisson marginals, $\boldsymbol{\alpha}=(\alpha_e,\,e\in\mathcal{E})$ for its dependence parameters, and $\mathcal{T}$ for its underlying tree.
We do not indicate a rooting for the tree: such a specification would be inconsequential given Theorem~\ref{th:Reversibility}.

The notion of subtree of a tree is relevant to the following corollary to Theorems~\ref{th:StoDynamics} and \ref{th:Reversibility}. Precisely, a subtree $\tau$ of a tree $\mathcal{T} = (\mathcal{V},\mathcal{E})$ is defined by the pair $(\mathcal{V}^{\tau},\mathcal{E}^{\tau})$, where $\mathcal{V}^{\tau}$ is a subset of vertices $\mathcal{V}^{\tau}\subseteq\mathcal{V}$ and $\mathcal{E}^{\tau}$ is a subset of edges $\mathcal{E}^{\tau}\subseteq\mathcal{E}$, such that $\tau$ is a tree itself. The translation to the MRF's dynamics of pruning $\mathcal{T}$ of $\tau$, meaning considering the residual tree comprising vertices $\mathcal{V}\backslash\mathcal{V}^{\tau}$, is examined.  

\begin{cor}
	\label{th:Pruning}
	Consider a tree $\mathcal{T}=(\mathcal{V},\mathcal{E})$ and a subtree $\tau=(\mathcal{V}^{\tau}, \mathcal{E}^{\tau})$ of $\mathcal{T}$. Let $\boldsymbol{N} = (N_v, \, v \in \mathcal{V})\sim\text{MPMRF}(\lambda,\boldsymbol{\alpha}, \mathcal{T})$ and $\boldsymbol{N}^{\tau}=(N^{\tau}_v,\, v\in\mathcal{V}^{\tau})\sim\text{MPMRF}(\lambda,\boldsymbol{\alpha}^{\tau}, \tau)$. Suppose $\alpha_e = \alpha^{\tau}_e$ for every $e\in\mathcal{E}^{\tau}$. Then, the subvector $(N_v,\, v\in\mathcal{V}^{\tau})$ is identically distributed as $\boldsymbol{N}^{\tau}$. In other words, neglecting the random variables associated to the set of vertices $\mathcal{V}\backslash\mathcal{V}^{\tau}$ does not affect the joint distribution of the random variables associated to vertices in $\mathcal{V}^{\tau}$.   
\end{cor}
\begin{proof} 
	Given Theorem~\ref{th:Reversibility}, one may choose any root without it affecting the joint distribution of $\boldsymbol{N}$. We choose a root $r^{\prime}\in\mathcal{V}^{\tau}$, and remark that every vertex $i\in\mathcal{V}^{\tau}$ has its parent in $\tau$, according to this rooting. Indeed, for $\tau$ to be a tree, all vertex on $\mathrm{path}(r^{\prime},i)$ must be in $\mathcal{V}^{\tau}$ as well for every $i\in\mathcal{V}^{\tau}$. Therefore, according to (\ref{eq:StoDynamics}), the stochastic representation of every $N_i$, $i\in\mathcal{V}^{\tau}$, is exactly the same whether the underlying rooted tree is $\mathcal{T}_{r^{\prime}}$ or $\tau_{r^{\prime}}$.
\end{proof} 

The result actually follows from the fixedness of marginals coupled to the global Markov property. Theorem~\ref{th:Reversibility} simply ensures the construction to remain well defined after the pruning, as one should evoke it to choose a root in the subtree of interest. 
Corollary~\ref{th:Pruning} conversely informs that conjuncting additional random variables, with respect to the representation given in (\ref{eq:StoDynamics}), to a MRF from the proposed family will not affect the dynamics already within that MRF, and thus will also yield a MRF from the proposed family. Being able to neglect or conjunct random variables, as provided by Corollary~\ref{th:Pruning}, is not a property shared by MRFs in general, due to their marginal distributions not being fixed. For a generic MRF, say $\boldsymbol{X}$, the simple existence of a random variable $X_u$ whose vertex connects to $v$, with $u,v\in\mathcal{V}$, influences the distribution of $X_v$. The result is also stronger than answering the question of marginalization of MRFs discussed in Chapter~4 of \cite{cressie2015statistics} : $\boldsymbol{N}^{\tau}$ is not only a MRF, it follows a distribution from $\mathbb{MPMRF}$ with the same parameters as for $\boldsymbol{N}$ for the non-pruned edges. 
We provide the following example to illustrate Corollary~\ref{th:Pruning}. 
\begin{ex} 
	\label{ex:TreesPruning}	
	Consider $\boldsymbol{N}$ and $\boldsymbol{N}^{\prime}$ respectively defined on trees $\mathcal{T}$ and $\mathcal{T}^{\prime}$ of Fig.~\ref{fig:TreesPruning}. Suppose both have the same parameter $\lambda$ for their marginal distributions, and the same dependence parameters $\alpha_e$ on their common edges. 
	Given Corollary~\ref{th:Pruning}, we have $(N_1,N_2,N_3,N_4,N_5) \stackrel{d}{=} (N_1^{\prime},N_2^{\prime}, N_3^{\prime}, N_4^{\prime}, N_5^{\prime})$.
\end{ex}

\begin{figure}[H]
	\centering
	\begin{tikzpicture}[every node/.style={text=Black, circle, draw = Maroon, inner sep=0mm, outer sep = 0mm, minimum size=3.5mm, fill = White}, node distance = 1.5mm, scale=0.1, thick]
		\node (1a) {\tiny $1$};
		\node [below=of 1a] (2a) {\tiny $2$};
		\node [below right=of 1a] (3a) {\tiny $3$};
		\node [below left =of 2a] (4a) {\tiny $4$};
		\node [below =of 2a] (5a) {\tiny $5$};
		\node [draw = Red!60, below left=of 1a] (6a) {\tiny $6$};
		\node [draw = Red!60, left=of 6a] (7a) {\tiny $7$};
		\node [draw = Red!60, below left=of 7a] (8a) {\tiny $8$};
		\node [draw = Red!60, below=of 7a] (9a) {\tiny $9$};
		\node [draw = Red!60, below right=of 3a] (10a) {\tiny $10$};
		\node [draw = Red!60, right=of 3a] (11a) {\tiny $11$};
		\node [draw = Red!60, below right=of 11a] (12a) {\tiny $12$};

		\draw (1a) -- (2a);
		\draw (1a) -- (3a);
		\draw (2a) -- (4a);
		\draw (2a) -- (5a);            
		\draw (1a) -- (6a);       
		\draw (6a) -- (7a);
		\draw (7a) -- (8a);
		\draw (7a) -- (9a);
		\draw (3a) -- (10a);            
		\draw (3a) -- (11a);       
		\draw (11a) -- (12a);

		\node [right = 27mm of 1a](1b) {\tiny $1$};
		\node [below=of 1b] (2b) {\tiny $2$};
		\node [below right=of 1b] (3b) {\tiny $3$};
		\node [below left =of 2b] (4b) {\tiny $4$};
		\node [below =of 2b] (5b) {\tiny $5$};

		\draw (1b) -- (2b);
		\draw (1b) -- (3b);
		\draw (2b) -- (4b);
		\draw (2b) -- (5b);            
		
		\node[draw=none, below = 11mm of 1a](t1){$\mathcal{T}^{\mathcolor{White}{\prime}}$};
		\node[draw=none, below= 11mm of 1b](t2){$\mathcal{T}^{\prime}$};
	\end{tikzpicture}
	\caption{Trees $\mathcal{T}$ and $\mathcal{T}^{\prime}$ of Example  \ref{ex:TreesPruning}.}
	\label{fig:TreesPruning}
\end{figure}

\section{Joint probability mass function and joint probability generating function}
\label{subsect:PMF-PGF}	

We examine the joint probability mass function and the joint probability generating function characterizing the family $\mathbb{MPMRF}$, both allowing a better understanding of the stochastic dynamics at play brought by the representation given in (\ref{eq:StoDynamics}). Deriving analytical expressions for both quantities is of no negligible value since it notably dispels the question of whether a valid distribution arises from the stochastic representation. Indeed, as explored in \cite{kaiser2000construction}, specifying conditional distributions does not suffice to ensure a valid joint probability mass (or density) function for a MRF. 
The stochastic representation (\ref{eq:StoDynamics}) in Theorem~\ref{th:StoDynamics}, however, provides a natural sequence of conditioning, that is, successively moving away from the root, which enables to derive well-defined analytic expressions for the joint probability mass function and joint probability generating function of $\boldsymbol{N}$. This is due to the stochastic construction defining itself on cliques of the tree, which we discuss after Theorem~\ref{th:JointPMF}. As a result, it always leads to a valid multivariate distribution.

For a count random vector $\boldsymbol{X} = (X_v,\,v\in\mathcal{V})$, let $p_{\boldsymbol{X}}$ denote its joint probability mass function, that is,
$p_{\boldsymbol{X}}(\boldsymbol{x})= \mathrm{Pr}\left(\bigcap_{v\in\mathcal{V}} \left\{X_v = x_v\right\}\right)$, $\boldsymbol{x} \in \mathbb{N}^d$.

\begin{theorem}[Joint probability mass function]
	\label{th:JointPMF} For a tree $\mathcal{T}=(\mathcal{V},\mathcal{E})$, consider $\boldsymbol{N} = (N_v, \, v \in \mathcal{V})\sim\text{MPMRF}(\lambda,\boldsymbol{\alpha}, \mathcal{T})$. Choose any root $r\in\mathcal{V}$; the joint probability mass function of ${\boldsymbol{N}}$ is then given by
	\begin{equation}
					p_{\boldsymbol{N}}(\boldsymbol{x}) =  \frac{\mathrm{e}^{-\lambda}\lambda^{x_r}}{x_r!} \prod_{v\in\mathcal{V}\backslash\{r\}} \sum_{k=0}^{\min(x_{\mathrm{pa}(v)},x_v)} \frac{{ \mathrm{e}}^{-\lambda(1-\alpha_{(\mathrm{pa}(v),v)})} (\lambda(1-\alpha_{(\mathrm{pa}(v),v)}))^{x_v-k}}{(x_v-k)!}
     \binom{x_{\mathrm{pa}(v)}}{k} \alpha_{(\mathrm{pa}(v),v)}^k (1-\alpha_{(\mathrm{pa}(v),v)})^{x_{\mathrm{pa}(v)}-k},
     \label{eq:jointpmf}
	\end{equation}
 for $\boldsymbol{x}\in\mathbb{N}^d$. 
\end{theorem}

\begin{proof}
 Given the global Markov property, assuming the chosen root $r$, and by successive conditioning, the joint probability mass function 
 of $\boldsymbol{N}$ admits the following representation: 
	\begin{equation}
		p_{\boldsymbol{N}}(\boldsymbol{x}) = p_{N_r}(x_r) \prod_{v\in\mathcal{V}\backslash\{r\}} p_{N_v|N_{\mathrm{pa}(v)} = x_{\mathrm{pa}(v)}}(x_v),
		\label{eq:MarkovProba}\\
	\end{equation}
	which, given (\ref{eq:StoDynamics}), may be rewritten using the convolution product as follows: 
	\begin{equation*}
		p_{\boldsymbol{N}}(\boldsymbol{x}) = p_{N_r}(x_r) \prod_{v\in\mathcal{V}\backslash\{r\}} \sum_{k=0}^{\min(x_v,x_{\mathrm{pa}(v)})} p_{L_v}(x_v-k) p_{\alpha_{(\mathrm{pa}(v),v)}\circ x_{\mathrm{pa}(v)}} (k),\quad\quad \boldsymbol{x} \in \mathbb{N}^d.
	\end{equation*}
	Given $L_v\sim$~Poisson$(\lambda(1-\alpha_{(\mathrm{pa}(v),v)}))$, $v\in\mathcal{V}$, and that $(\alpha_{(\mathrm{pa}(v),v)}\circ x_{\mathrm{pa}(v)})$ follows a binomial distribution of size parameter $x_{\mathrm{pa}(v)}$ and probability parameter $\alpha_{(\mathrm{pa}(v),v)}$ due to (\ref{eq:BinomThinning}), we obtain the desired result. 
\end{proof}

One can recognize in (\ref{eq:MarkovProba}) the expression of the joint probability mass function associated to a Gibbs distribution factorizing on $\mathcal{T}$, whose forthcoming definition is provided with respect to \cite{koller2009probabilistic}. 
\begin{deff}[Gibbs distribution]
	Let $\mathcal{V}_1,\ldots, \mathcal{V}_m$ be subsets of $\mathcal{V}$, $m\in\mathbb{N}^*$, and $\phi_1,\ldots,\phi_m$ be some functions $\phi_i : \mathbb{R}^{|\mathcal{V}_i|}\mapsto\mathbb{R}$, $i\in\{1,\ldots,m\}$. A vector of random variables $\boldsymbol{X}=(X_v,\, v\in\mathcal{V})$ follows a Gibbs distribution if its joint density (or joint probability mass function) may be expressed as 
	\begin{equation*}
		f_{\boldsymbol{X}}(\boldsymbol{x}) = \frac{1}{Z} \prod_{i=1}^{m}\phi_i((x_v,\, v\in\mathcal{V}_i)),
	\end{equation*}   
	where $Z$ is a normalizing constant. A Gibbs distribution factorizes on a graph if $\,\mathcal{V}_1,\ldots,\mathcal{V}_m$ are all cliques of that graph. Note that, on a tree, cliques are any single vertex or pair of vertices connected by an edge. 
\end{deff}
Gibbs distributions are tied to MRFs by the Hammersley-Clifford theorem, stating that any MRF on a graph $\mathcal{G}$ follows a Gibbs distribution factorizing on $\mathcal{G}$. Therefore, using the global Markov property characterizing MRFs to derive (\ref{eq:MarkovProba}) ensured obtaining a Gibbs distribution. Also, notice that the resulting joint probability mass function in (\ref{eq:jointpmf}) does not have the intractable normalizing factor $Z$ typical to Gibbs distributions. In our context, factors $\phi_1,\ldots,\phi_m$ in ($\ref{eq:MarkovProba}$) are probability mass functions, namely $\phi_{r}(x_r) = p_{N_r}(x_r)$ and $\phi_{v}(x_{\mathrm{pa}(v)}, x_v) = p_{N_v|N_{\mathrm{pa}(v)} = x_{\mathrm{pa}(v)}}(x_v)$, $v\in\mathcal{V}\backslash\{r\}$. This is feasible because the stochastic representation ($\ref{eq:StoDynamics}$) in Theorem~\ref{th:StoDynamics} defines itself on cliques of the tree (a vertex and its parent form a clique).  Such a feature is incidentally tied to the fixed marginal distributions; it was indeed envisioned by the author of \cite{pickard1980unilateral} to attest the existence of lattice-structured stationary MRFs.
The absence of intractable normalizing factor in (\ref{eq:jointpmf}) eases its computation: as for the derivation of marginals, one directly has the result at hand and does not have to resort to message passing algorithms. For more insight on Gibbs distributions, one may refer to Chapter 4 of \cite{koller2009probabilistic}.

The stochastic construction of (\ref{eq:StoDynamics}) in Theorem~\ref{th:StoDynamics} allows to derive the joint probability generating function of $\boldsymbol{N}$, that is, $\mathcal{P}_{\boldsymbol{N}}(\boldsymbol{t}) = \mathrm{E}[t_1^{N_1} \dots t_d^{N_d}]$. The next theorem offers a recursive expression of the joint probability generating function of $\boldsymbol{N}$. 
This proves useful for numerical applications using Algorithm~\ref{algo:fftM} in Section~\ref{sect:Sum}. 
For more details on joint probability generating functions and their properties, see Appendix~A in \cite{axelrod2015branching}, and Chapter~3 of \cite{flajolet2009analytic}.
For notation purposes, let $v\mathrm{dsc}(v)$ abbreviate the set $\{v\}\cup \mathrm{dsc}(v)$ and, for a vector $\boldsymbol{t}\in[-1,1]^d$, let $\boldsymbol{t}_{v\mathrm{dsc}(v)}= \left(t_j,\, j\in v\mathrm{dsc}(v)\right)$, $v\in\mathcal{V}$; evidently, $\boldsymbol{t}_{r\mathrm{dsc}(r)} = \boldsymbol{t}$ for a root $r\in\mathcal{V}$.

\begin{theorem}[Joint probability generating function]
	For a tree $\mathcal{T}=(\mathcal{V},\mathcal{E})$, consider $\boldsymbol{N} = (N_v, \, v \in \mathcal{V})\sim\text{MPMRF}(\lambda,\boldsymbol{\alpha}, \mathcal{T})$. Choose any root $r\in\mathcal{V}$; the joint probability generating function of ${\boldsymbol{N}}$ is then given by
	\begin{equation}
		\mathcal{P}_{\boldsymbol{N}}({\boldsymbol{t}}) = \prod_{v \in \mathcal{V}} \mathrm{e}^{\lambda(1-\alpha_{(\mathrm{pa}(v),v)})\left(\eta^{\mathcal{T}_r}_v(\boldsymbol{t}_{v\mathrm{dsc}(v)})-1\right)},
		\quad \boldsymbol{t} \in [-1,1]^d, \label{eq:JointPGF}
	\end{equation}
	where $\eta^{\mathcal{T}_r}_v(\boldsymbol{t}_{v\mathrm{dsc}(v)})$, $v\in\mathcal{V}$, is a polynomial recursively defined, starting from the leaves of $\mathcal{T}_r$, given by
	\begin{equation}
		\eta^{\mathcal{T}_r}_v(\boldsymbol{t}_{v\mathrm{dsc}(v)})= t_v \prod_{j \in \mathrm{ch}(v)}\left(1 - \alpha_{(v,j)} + \alpha_{(v,j)} 	\eta^{\mathcal{T}_r}_j(\boldsymbol{t}_{j\mathrm{dsc}(j)})\right). \label{eq:jointpgf-h}
	\end{equation}
	If $\mathrm{ch}(v)$ is empty, then $\eta^{\mathcal{T}_r}_v(\boldsymbol{t}_{v\mathrm{dsc}(v)}) = t_v$.  
	\label{th:JointPGF}
\end{theorem}

\begin{proof}
	Let $\mathcal{H}_v(\boldsymbol{t}_{v\mathrm{dsc}(v)}) = \mathrm{E}\left[\left.t_v^{N_v} \prod_{j\in \mathrm{dsc}(v)} t_j^{N_j}\right| N_{\mathrm{pa}(v)}\right]$, $v\in\mathcal{V}$. Note that $\mathcal{P}_{\boldsymbol{N}}(\boldsymbol{t}) = \mathcal{H}_r(\boldsymbol{t}_{r\mathrm{dsc(r)}})$ with the convention $\alpha_{(\mathrm{pa}(r),r)} = 0$. Indeed, conditioning on an imaginary $N_{\mathrm{pa}(r)}$ provides no information given (\ref{eq:StoDynamics}) and the global Markov property. First, we reexpress $\mathcal{H}_v$ recursively on the values of $\{\mathcal{H}_j,\, j\in\mathrm{ch}(v)\}$. Reorganizing the product on descendant vertices, we have
	\begin{equation}
		\mathcal{H}_v\left(\boldsymbol{t}_{v\mathrm{dsc}(v)}\right) = \mathrm{E}\left[\left. t_v^{N_v} \prod_{j\in\mathrm{ch}(v)} \left(t_j^{{N_j}}\prod_{i\in\mathrm{dsc}(j)}{t_i^{N_i}}\right)\right| N_{\mathrm{pa}(v)} \right] ,\quad \boldsymbol{t}\in[-1,1]^d.
		\label{eq:JointPGF-new-1}
	\end{equation}
	Due to $\boldsymbol{N}$ satisfying the global Markov property, $N_{\mathrm{pa}(v)}$ and $N_{j}$, $j\in\mathrm{ch}(v)$, are all conditionally independent of one another knowing $N_v$. Therefore, by conditioning on $N_v$, (\ref{eq:JointPGF-new-1}) becomes
	\begin{align}
		\mathcal{H}_v\left(\boldsymbol{t}_{v\mathrm{dsc}(v)}\right) &= \mathrm{E}\left[ \mathrm{E}\left[\left. t_v^{N_v} \prod_{j\in\mathrm{ch}(v)} \left(t_j^{{N_j}}\prod_{i\in\mathrm{dsc}(j)}{t_i^{N_i}}\right)\right| N_{\mathrm{pa}(v)}, N_v\right]\right] = \mathrm{E}\left[ \left. t_v^{N_v} \prod_{j\in\mathrm{ch}(v)}\mathrm{E}\left[\left.t_j^{N_j}\prod_{i\in\mathrm{dsc}(j)}{t_i^{N_i}}\right| N_v \right] \right| N_{\mathrm{pa}(v)}\right] \notag\\
		&= \mathrm{E}\left[ \left. t_v^{N_v} \prod_{j\in\mathrm{ch}(v)} \mathcal{H}_j(\boldsymbol{t}_{j\mathrm{dsc}(j)}) \right| N_{\mathrm{pa}(v)}\right] ,\quad \boldsymbol{t}\in[-1,1]^d.
		\label{eq:JointPGF-new-2}
	\end{align}
	Next, we proceed by induction to prove that $\mathcal{H}_v\left(\boldsymbol{t}_{v\mathrm{dsc}(v)}\right)$, for any $v\in\mathcal{V}$, admits the following representation:
	{\setlength{\abovedisplayskip}{3pt}
		\setlength{\belowdisplayskip}{3pt}
		\begin{equation}
			{\scalebox{0.92}{
					\begin{minipage}{\linewidth}
						\begin{equation*}
							\mathcal{H}_v\left(\boldsymbol{t}_{v\mathrm{dsc}(v)}\right) = \left(1-\alpha_{(\mathrm{pa}(v),v)} + \alpha_{(\mathrm{pa}(v),v)}	\eta^{\mathcal{T}_r}_v(\boldsymbol{t}_{v\mathrm{dsc}(v)})\right)^{N_{\mathrm{pa}(v)}} \prod_{j\in (\{v\}\cup\mathrm{dsc}(v)\})}\mathrm{e}^{\lambda(1-\alpha_{(\mathrm{pa}(j),j)})\left(	\eta^{\mathcal{T}_r}_j(\boldsymbol{t}_{j\mathrm{dsc}(j)})-1\right)},
							\label{eq:JointPGF-new-3}
			\end{equation*}\end{minipage}}}
	\end{equation}}\\
	for $\boldsymbol{t}\in[-1,1]^d$. We take the leaves of the rooted tree $\mathcal{T}_r$ as starting points. Let $\ell$ be any leaf. Then,
	\begin{equation}
		\mathcal{H}_{\ell}\left(\boldsymbol{t}_{\ell\mathrm{dsc}(\ell)}\right) = \mathrm{E}\left[\left. t_{\ell}^{N_{\ell}}\right| N_{\mathrm{pa}(\ell)}\right] = \mathrm{E}\left[\left.t_{\ell}^{\alpha_{(\mathrm{pa}(\ell),\ell)}\circ N_{\mathrm{pa}(\ell)}}\right| N_{\mathrm{pa}(\ell)}\right]\times \mathrm{E}\left[t_{\ell}^{L_{\ell}}\right],\quad \boldsymbol{t} \in [-1,1]^d.
		\label{eq:JointPGF-new-4}
	\end{equation}
	The last equality in (\ref{eq:JointPGF-new-4}) is obtained by the independence of $L_{\ell}$ and $N_{\mathrm{pa}(\ell)}$. Given $L_{\ell}$ follows a Poisson distribution of parameter $\lambda(1-\alpha_{(\mathrm{pa}(\ell),\ell)})$ and $(\alpha_{(\mathrm{pa}(\ell),\ell)}\circ N_{\mathrm{pa}(\ell)}|N_{\mathrm{pa}(\ell)})$ follows a binomial distribution of parameters $N_{\mathrm{pa}(\ell)}$ and $\alpha_{(\mathrm{pa}(\ell),\ell)}$, we obtain
	\begin{align}
		\mathcal{H}_{\ell}\left(\boldsymbol{t}_{\ell\mathrm{dsc}(\ell)}\right) = \left(1-\alpha_{(\mathrm{pa}(\ell),\ell)} + \alpha_{(\mathrm{pa}(\ell),\ell)}t_{l}\right)^{N_{\mathrm{pa}(\ell)}} \times \mathrm{e}^{\lambda(1-\alpha_{(\mathrm{pa}(\ell),\ell)})(t_{\ell}-1)},\quad \boldsymbol{t} \in [-1,1]^d.
		\label{eq:JointPGF-new-5}
	\end{align}
	We note that (\ref{eq:JointPGF-new-5}) agrees with (\ref{eq:JointPGF-new-3}) since $\eta_{\ell}^{\mathcal{T}_r}(\boldsymbol{t}_{\ell\mathrm{dsc}(\ell)}) = t_{\ell}$ given $\ell$ is a leaf of $\mathcal{T}_r$. Now, for any vertex $k\in\mathcal{V}$, assume (\ref{eq:JointPGF-new-3}) holds for any $j\in\mathrm{ch}(k)$. Then, (\ref{eq:JointPGF-new-2}) becomes 
						\begin{equation}
							\mathcal{H}_k\left(\boldsymbol{t}_{k\mathrm{dsc}(k)}\right) 
							= \mathrm{E}\left[\left. t_{k}^{N_{k}} \prod_{j \in \mathrm{ch}(k)} \left(1-\alpha_{(k,j)} + \alpha_{(k,j)}	\eta^{\mathcal{T}_r}_j(\boldsymbol{t}_{j\mathrm{dsc}(j)})\right)^{N_{k}} \times \mathrm{e}^{\lambda(1-\alpha_{(k,j)}) \left(\eta^{\mathcal{T}_r}_j\left(\boldsymbol{t}_{j\mathrm{dsc}(j)}\right)-1\right)} \right|N_{\mathrm{pa}(k)} \right]. \label{eq:JointPGF-new-6}
			\end{equation}
	Given the independence of $L_k$ and $N_{\mathrm{pa}(k)}$, (\ref{eq:JointPGF-new-6}) is rewritten as 
	\begin{align}
							\mathcal{H}_k\left(\boldsymbol{t}_{k\mathrm{dsc}(k)}\right) &= \mathrm{E}\left[\left.\left(t_{k}\prod_{j \in \mathrm{ch}(k)} \left(1-\alpha_{(k,j)} + \alpha_{(k,j)}	\eta^{\mathcal{T}_r}_j(\boldsymbol{t}_{j\mathrm{dsc}(j)})\right)\right)^{\alpha_{(\mathrm{pa}(k),k)}\circ N_{\mathrm{pa}(k)}} \right|N_{\mathrm{pa}(k)}\right]\notag\\
			&\quad\times  \mathrm{E}\left[\left(t_{k}\prod_{j \in \mathrm{ch}(k)} \left(1-\alpha_{(k,j)} + \alpha_{(k,j)}	\eta^{\mathcal{T}_r}_j(\boldsymbol{t}_{j\mathrm{dsc}(j)})\right)\right)^{L_{k}} \right]\prod_{j\in \mathrm{dsc}(k)}\mathrm{e}^{\lambda(1-\alpha_{(\mathrm{pa}(j),j)})\left(\eta^{\mathcal{T}_r}_j(\boldsymbol{t}_{j\mathrm{dsc}(j)})-1\right)},
							\label{eq:JointPGF-new-7}
			\end{align}
	for $\boldsymbol{t} \in [-1,1]^d$. Since expectations in (\ref{eq:JointPGF-new-7}) are probability generating functions of $\left(\alpha_{(\mathrm{pa}(k),k)}\circ N_{\mathrm{pa}(k)}|N_{\mathrm{pa}(k)}\right)$ and $L_k$ respectively, we obtain
		\begin{align}
		\mathcal{H}_k\left(\boldsymbol{t}_{k\mathrm{dsc}(k)}\right) = &\left(1-\alpha_{(\mathrm{pa}(k),k)} + \alpha_{(\mathrm{pa}(k),k)}\left(t_{k}\prod_{j \in \mathrm{ch}(k)} \left(1-\alpha_{(k,j)} + \alpha_{(k,j)}	\eta^{\mathcal{T}_r}_j(\boldsymbol{t}_{j\mathrm{dsc}(j)})\right)\right)\right)^{N_{\mathrm{pa}(k)}}\notag\\
			&\quad\times  \mathrm{e}^{\lambda(1-\alpha_{(\mathrm{pa}(k),k)})\left(t_{k}\prod_{j \in \mathrm{ch}(k)} (1-\alpha_{(k,j)} + \alpha_{(k,j)}	\eta^{\mathcal{T}_r}_j(\boldsymbol{t}_{j\mathrm{dsc}(j)}))-1\right)}\prod_{j\in \mathrm{dsc}(k)}\mathrm{e}^{\lambda(1-\alpha_{(\mathrm{pa}(j),j)})\left(\eta^{\mathcal{T}_r}_j(\boldsymbol{t}_{j\mathrm{dsc}(j)})-1\right)},
								\label{eq:JointPGF-new-8}
				\end{align}
		for $\boldsymbol{t} \in [-1,1]^d$.
	 Relation in (\ref{eq:JointPGF-new-3}) is retrieved from (\ref{eq:JointPGF-new-8}) by recognizing $\eta_{k}^{\mathcal{T}_r}(\boldsymbol{t}_{k\mathrm{dsc}(k)})$ as given by (\ref{eq:jointpgf-h}). This inductively proves (\ref{eq:JointPGF-new-3}).

	Since $\mathcal{P}_{\boldsymbol{N}}(\boldsymbol{t})=\mathcal{H}_{r}(\boldsymbol{t}_{r\mathrm{dsc}(r)})$ with the convention $\alpha_{(\mathrm{pa}(r),r)}=0$, we then derive the expression of the joint probability generating function from (\ref{eq:JointPGF-new-3}) :
	\begin{align*}
		\mathcal{P}_{\boldsymbol{N}}(\boldsymbol{t}) &= \left(1-\alpha_{(\mathrm{pa}(r),r)} + \alpha_{(\mathrm{pa}(r),r)}	\eta^{\mathcal{T}_r}_r(\boldsymbol{t}_{r\mathrm{dsc}(r)})\right)^{N_{\mathrm{pa}(r)}} \prod_{j\in (\{r\}\cup\mathrm{dsc}(r))}\mathrm{e}^{\lambda(1-\alpha_{(\mathrm{pa}(j),j)})\left(	\eta^{\mathcal{T}_r}_j(\boldsymbol{t}_{j\mathrm{dsc}(j)})-1\right)}\notag\\
		&= (1)^{N_{\mathrm{pa}(r)}} \prod_{j\in \mathcal{V}}\mathrm{e}^{\lambda(1-\alpha_{(\mathrm{pa}(j),j)})\left(	\eta^{\mathcal{T}_r}_j(\boldsymbol{t}_{j\mathrm{dsc}(j)})-1\right)} = \prod_{j\in \mathcal{V}}\mathrm{e}^{\lambda(1-\alpha_{(\mathrm{pa}(j),j)})\left(	\eta^{\mathcal{T}_r}_j(\boldsymbol{t}_{j\mathrm{dsc}(j)})-1\right)},\quad \boldsymbol{t}\in[-1,1],
	\end{align*}
	whichever the value of that imaginary $N_{\mathrm{pa}(r)}$. This grants the desired result. 
\end{proof}

Note that $\eta_{v}^{\mathcal{T}_r}$ is itself a joint probability generating function for any $v\in\mathcal{V}$. More precisely, it is the joint probability generating function of a multitype time-varying Galton-Watson process with each vertex being its own type and having the offspring distribution be binomial with success probability $\alpha$ and size parameter the number of children of the vertex. In this regard, one may consult Sections 1.13.1, 1.13.5 and 2.3 of \cite{harris1963theory}.

Let $\boldsymbol{G}_{v}^{\mathcal{T}_r}= (G_{v,j}^{\mathcal{T}_r}, {\scriptstyle j\in(\{v\}\cup\mathrm{dsc}(v))})$ be a vector of random variables whose joint probability generating function is given by $\eta_{v}^{\mathcal{T}_r}(\boldsymbol{t}_{v\mathrm{dsc}(v)})$.
A stochastic interpretation of $\boldsymbol{G}_{v}^{\mathcal{T}_r}$ is as follows: an event triggered at vertex $v$ has the possibility to propagate through its descendants following the edges of the rooted tree $\mathcal{T}_r$, thus creating the looks of a splatter around $v$. From (\ref{eq:JointPGF}), we identify that $\boldsymbol{N}$ is stochastically rewritten as the componentwise sum of $d$ independent random vectors following a multivariate compound Poisson distribution, one associated to each vertex in $\mathcal{V}$. Moreover, one notices that the counting distribution parameter of each compound Poisson distribution is $\lambda(1-\alpha_{(\mathrm{pa}(v),v)})$, with  $\alpha_{(\mathrm{pa}(r),r)}=0$, and so the counting distributions may be describing the behavior of $\boldsymbol{L}$ defined as in Theorem \ref{th:StoDynamics}. Therefore, another stochastic construction for $\boldsymbol{N}$ is 
\begin{equation}
	\boldsymbol{N} = \sum_{v\in\mathcal{V}}\sum_{k=1}^{L_v} \boldsymbol{G}_{v,(k)}^{\mathcal{T}_r},
	\label{eq:AlternateConstr}
\end{equation}
with sums taken componentwise, and $\{\boldsymbol{G}_{v,(k)}^{\mathcal{T}_r}, k\in\mathbb{N}^*\}$ being a sequence of independent vectors of random variables having the same joint distribution as $\boldsymbol{G}_{v}^{\mathcal{T}_r}$, $v\in\mathcal{V}$. This represents summing number of events according to their origin in terms of propagation.

Expressing in superscript the rooted tree on which $\eta_{v}^{\mathcal{T}_r}$ is defined may seem redundant since $\mathcal{T}_r$ is clearly specified in the formulation of Theorem~\ref{th:JointPGF}.
 It will provide clarity later on, notably for Section~\ref{sect:ExpAlloc}.
Incidently, the same development to derive the expression of $\mathcal{P}_{\boldsymbol{N}}(\boldsymbol{t})$ can be done for any other root $r\in \mathcal{V}$. The resulting joint probability generating function is the same, only expressed differently. Indeed, the choice of the root has no stochastic incidence, as stated in Theorem~\ref{th:Reversibility}; choosing a root only serves to define parents and children for the construction of the joint probability generating function through Theorem~\ref{th:JointPGF}. For illustration purposes, we derive the joint probability generating function of a three-vertex series tree with root $r\in\{1, 2, 3\}$:
\begin{equation}
 \mathcal{P}_{\boldsymbol{N}}(\boldsymbol{t})=
							\begin{cases}
								\mathrm{e}^{\lambda(t_1(1-\alpha_{(1,2)}+\alpha_{(1,2)}t_2(1-\alpha_{(2,3)}+\alpha_{(2,3)}t_3))-1)}
        \times\mathrm{e}^{\lambda(1-\alpha_{(1,2)})(t_2(1-\alpha_{(2,3)}+\alpha_{(2,3)}t_3)-1)}
        \times\mathrm{e}^{\lambda(1-\alpha_{(2,3)})(t_3-1)},&
        {r = 1;} \\ 
								\mathrm{e}^{\lambda(t_2(1-\alpha_{(1,2)}+\alpha_{(1,2)}t_1)(1-\alpha_{(2,3)}+\alpha_{(2,3)}t_3)-1)}
        \times\mathrm{e}^{\lambda(1-\alpha_{(1,2)})(t_1-1)}
        \times\mathrm{e}^{\lambda(1-\alpha_{(2,3)})(t_3-1)},&{ r = 2;} \\ 
								\mathrm{e}^{\lambda(t_3(1-\alpha_{(2,3)}+\alpha_{(2,3)}t_2(1-\alpha_{(1,2)}+\alpha_{(1,2)}t_1))-1)}
								\times\mathrm{e}^{\lambda(1-\alpha_{(2,3)})(t_2(1-\alpha_{(1,2)}+\alpha_{(1,2)}t_1)-1)}
        \times\mathrm{e}^{\lambda(1-\alpha_{(1,2)})(t_1-1)},& { r = 3},
								\label{eq:SameRoot}
        \end{cases}
	\end{equation}
	for $\boldsymbol{t}\in[-1,1]^3$. In appearance, the three joint probability generating functions in \eqref{eq:SameRoot} seem distinct; however, when developed, they all lead to the same following expression:
    \begin{align*}
					\mathcal{P}_{\boldsymbol{N}}(\boldsymbol{t}) &= \mathrm{e}^{\lambda(\alpha_{(1,2)}\alpha_{(2,3)}t_1t_2t_3 + \alpha_{(1,2)}(1-\alpha_{(2,3)})t_1t_2 
						+ (1-\alpha_{(1,2)})\alpha_{(2,3)} t_2t_3 + (1-\alpha_{(1,2)})t_1 + (1-\alpha_{(1,2)})(1-\alpha_{(2,3)})t_2 + (1-\alpha_{(2,3)})t_3 + \alpha_{(1,2)} + \alpha_{(2,3)} - 3)}. 
	\end{align*}

The expression of the joint probability generating function derived in Theorem~\ref{th:JointPGF} shows that independence and comonotonicity are two limit cases of our proposed family of multivariate distributions. If $\alpha_e = 0$ for all $e \in \mathcal{E}$, then $\boldsymbol{N}$ is a vector of independent, identically distributed Poisson random variables. Indeed, from (\ref{eq:jointpgf-h}), we obtain $\eta^{\mathcal{T}_r}_v(\boldsymbol{t}_{v\mathrm{dsc}(v)}) = t_v$, for every $v\in\mathcal{V}$, leading (\ref{eq:JointPGF}) to become $
\mathcal{P}_{\boldsymbol{N}}(\boldsymbol{t}) = \prod_{v\in\mathcal{V}}\mathrm{e}^{\lambda(t_v-1)}$, $\boldsymbol{t} \in [-1,1]^d$.
If $\alpha_e = 1$ for all $e \in \mathcal{E}$, then $\boldsymbol{N}$ is a vector of comonotonic Poisson random variables, as its joint probability generating function, provided Theorem~\ref{th:JointPGF}, is given by
$\mathcal{P}_{\boldsymbol{N}}(\boldsymbol{t}) = \mathrm{e}^{\lambda\left(\prod_{v\in\mathcal{V}} t_v - 1\right)}$,  $\boldsymbol{t} \in [-1,1]^d.$ 
See Section 1.9 in \cite{denuit2006} for details on comonotonicity. Capturing both independence and comonotonicity shows that the proposed family encompasses a wide spectrum of positive dependence. Negative dependence schemes cannot be captured with our construction, as exhibited in the next section.

\section{Covariance}
\label{subsect:Covariance}

Since dependence flows along the edges of the underlying tree of $\boldsymbol{N}$, we expect the covariance between two components of $\boldsymbol{N}$ to be described by the path from one component to the other. This is shown in the following theorem. 

\begin{theorem} [Covariance]
	\label{th:Covariance}
	For a tree $\mathcal{T}=(\mathcal{V},\mathcal{E})$, consider $\boldsymbol{N} = (N_v, \, v \in \mathcal{V})\sim\text{MPMRF}(\lambda,\boldsymbol{\alpha}, \mathcal{T})$. The covariance between two components of $\boldsymbol{N}$ is given by
	\begin{equation}
		\mathrm{Cov}\left(N_{u}, N_{v}\right) = \lambda \prod_{e \in \mathrm{path}(u,v)} \alpha_{e}, \quad u, v \in \mathcal{V}. 
		\label{eq:CovNj1Nj2}
	\end{equation}
\end{theorem}

\begin{proof}
	First, we find the expression of the covariance between two vertices connected by an edge. We thus suppose $u=\mathrm{pa}(v)$. From (\ref{eq:StoDynamics}), we have 
	\begin{equation} 
		\mathrm{Cov}(N_{u},N_{v}) = \mathrm{Cov}(N_{u}, \alpha_{(u,v)}\circ N_{u}) + \mathrm{Cov}(N_{u}, L_{v}),
		\quad (u,v) \in \mathcal{E}.
		\label{eq:CovNj1Nj2Edge-1}
	\end{equation}
	The second term in \eqref{eq:CovNj1Nj2Edge-1} is null given the independence of $N_{u}$ and $L_v$. By Theorem~\hyperref[th:PropertyBinThinOp]{\ref{th:PropertyBinThinOp}(c)}, the remaining term
	$\mathrm{Cov}(N_{u}, \alpha_{(u,v)}\circ N_{u})$ becomes $\lambda\alpha_{(u,v)}
	$.
	Hence, 
	\begin{equation}
		\mathrm{Cov}(N_u,N_v) = \lambda\alpha_{(u,v)}, \quad (u,v) \in \mathcal{E}.
		\label{eq:CovNj1Nj2Edge-Res}
	\end{equation}
	This corresponds to (\ref{eq:CovNj1Nj2}) when $\mathrm{path}(u,v)$ contains only one edge.

	Next, we look at the case where the path between the two vertices of interest has more than one edge. To facilitate the explanation, we change the labels of the vertices by adding meaningful subscripts. Let us then denote by $\{\ell_i, \; i \in \{1,\ldots,k\}\}$ the set of successive vertices on $\mathrm{path}(u,v)$, where $k$ is the number of such vertices, and $\ell_1=u$ and $\ell_k=v$. This allows to iteratively find the expression of the covariance between vertices $N_u$ and $N_v$, now written
	\begin{equation*}
		\mathrm{Cov}(N_{u},N_{v}) = \mathrm{Cov}(N_{\ell_1},N_{\ell_k}), \quad u,v \in \mathcal{V}.
	\end{equation*}
	We change the root of the underlying tree to $\ell_1$. Theorem~\ref{th:Reversibility} ensures this will not affect the result. We then have $\ell_i = \mathrm{pa}(\ell_{i+1})$ for every $i\in\{1,\ldots,k-1\}$. 
	We condition on the value taken by $N_{\ell_{k\text{-}1}}$, the random variable associated to the first intermediary vertex on the path from $\ell_k$ to $\ell_1$, and we find 
	\begin{equation}
		\mathrm{Cov}(N_{\ell_1}, N_{\ell_k}) = \mathrm{Cov}(\mathrm{E}[N_{\ell_1}|N_{\ell_{k\text{-}1}}],\mathrm{E}[N_{\ell_k}|N_{\ell_{k\text{-}1}}]) + \mathrm{E}[\mathrm{Cov}(N_{\ell_1}, N_{\ell_k}|N_{\ell_{k\text{-}1}})].
		\label{eq:CovNj1Nj2Noedge-2}
	\end{equation}
	
	The second term in (\ref{eq:CovNj1Nj2Noedge-2}) is equal to zero from the conditional independence given by the global Markov property. Then, from (\ref{eq:StoDynamics}), 
	\begin{align}
		\mathrm{Cov}(N_{\ell_1}, N_{\ell_k}) &= \mathrm{Cov}\left(\mathrm{E}[N_{\ell_1}|N_{\ell_{k\text{-}1}}],\mathrm{E}[\alpha_{(\ell_{k\text{-}1},\ell_k)}\circ N_{\ell_{k\text{-}1}}+L_{\ell_k}|N_{\ell_{k\text{-}1}}]\right) \notag \\
		&= \mathrm{Cov}\left(\mathrm{E}[N_{\ell_1}|N_{\ell_{k\text{-}1}}],\mathrm{E}[\alpha_{(\ell_{k\text{-}1},\ell_k)}\circ N_{\ell_{k\text{-}1}}|N_{\ell_{k\text{-}1}}]\right) + \mathrm{Cov}\left(\mathrm{E}[N_{\ell_1}|N_{\ell_{k\text{-}1}}],\mathrm{E}[L_{\ell_k}|N_{\ell_{k\text{-}1}}]\right).
		\label{eq:CovNj1Nj2Noedge-3}
	\end{align}
	
	The second term in (\ref{eq:CovNj1Nj2Noedge-3}) is null by the independence of $L_{\ell_{k}}$ with $N_{\ell_{k\text{-}1}}$. Consequently, 
	\begin{align*}
		\mathrm{Cov}(N_{\ell_1}, N_{\ell_k}) 
		&= \mathrm{Cov}\left(\mathrm{E}[N_{\ell_1}|N_{\ell_{k\text{-}1}}],\mathrm{E}\left[\left.\sum_{i=1}^{N_{\ell_{k\text{-}1}}} I_i^{(\alpha_{(\ell_{k\text{-}1},\ell_k)})} \right|N_{\ell_{k\text{-}1}}\right]\right)= \alpha_{(\ell_{k\text{-}1},\ell_{k})} \mathrm{Cov}\left(\mathrm{E}[N_{\ell_1}|N_{\ell_{k\text{-}1}}],N_{\ell_{k\text{-}1}}\right) = \alpha_{(\ell_{k\text{-}1},\ell_{k})}\mathrm{Cov}\left(N_{\ell_1},N_{\ell_{k\text{-}1}}\right).
	\end{align*}
	
	Repeating the same rationale iteratively with vertex $\ell_{k-2}$, and so on, up to vertex $\ell_{2}$, we find
	\begin{align*}
		\mathrm{Cov}(N_{\ell_1}, N_{\ell_k}) &=\alpha_{(\ell_{k\text{-}2},\ell_{k\text{-}1})}\alpha_{(\ell_{k\text{-}1},\ell_k)} \mathrm{Cov}(N_{\ell_1}, N_{\ell_{k\text{-}2}})=  \alpha_{(\ell_{k\text{-}3},\ell_{k\text{-}2})}\alpha_{(\ell_{k\text{-}2},\ell_{k\text{-}1})}\alpha_{(\ell_{k\text{-}1},\ell_k)}\mathrm{Cov}(N_{\ell_1}, N_{\ell_{k\text{-}3}})=\ldots\notag\\
		&= \alpha_{(\ell_2,\ell_3)}
		 \dots 
		\alpha_{(\ell_{k\text{-}2},\ell_{k\text{-}1})} \alpha_{(\ell_{k\text{-}1},\ell_{k})}\mathrm{Cov}(N_{\ell_{1}}, N_{\ell_2}) 
		 ,
	\end{align*}
	and given (\ref{eq:CovNj1Nj2Edge-Res}), 
	\begin{equation*}
		\mathrm{Cov}(N_{\ell_1}, N_{\ell_k}) = \lambda\alpha_{(\ell_1,\ell_2)}\alpha_{(\ell_2,\ell_3)}  \dots  \alpha_{(\ell_{k\text{-}2},\ell_{k\text{-}1})}   \alpha_{(\ell_{k\text{-}1},\ell_{k})} .
	\end{equation*}
	This grants the desired result when we revert back to the original subscripts $u,v \in \mathcal{V}$. 
\end{proof}

From Theorem \ref{th:Covariance} and the fact that each component of $\boldsymbol{N}$ is Poisson distributed with variance $\lambda$, as stated by Theorem \ref{th:StoDynamics}, the Pearson correlation coefficient between two components of $\boldsymbol{N}$ is given by
\begin{equation}
	\rho_P(N_{u}, N_{v}) = \prod_{e \in \mathrm{path}(u,v)} \alpha_{e}, \quad u, v \in \mathcal{V}. 
	\label{eq:rhoNj1Nj2}
\end{equation}

Since $\alpha_e \in [0,1]$ for all $e \in \mathcal{E}$, equation (\ref{eq:rhoNj1Nj2}) shows that only positive dependence ties the random variables within $\boldsymbol{N}$. It also shows that correlation exponentially decays as the path between two vertices becomes longer, meaning $\rho_P(N_u,N_v) \geq \rho_P(N_u,N_w)$ if $\mathrm{path}(u,v)\subseteq\mathrm{path}(u,w)$, $u,v,w \in \mathcal{V}$.

Computing the variance-covariance matrix of $\boldsymbol{N}$ by relying on Theorem \ref{th:Covariance} relates closely to computing the distance matrix of a weighted graph, a problem prominently studied in the graph theory literature as the "All-Pair Shortest Path" (APSP) problem.
The distance matrix of a graph, denoted $\boldsymbol{D} = (D_{ij},\; i\times j \in \mathcal{V}\times\mathcal{V})$, catalogs the distance between any pair of its vertices by summing the weights associated to the edges constituting the shortest path from one to the other. By choosing $\delta_{e}= -\ln\alpha_{e}$ as weight associated to edge $e$, $e\in\mathcal{E}$, the relation between the two problems becomes apparent: from (\ref{eq:CovNj1Nj2}), we have
\begin{equation}
	\mathrm{Cov}(N_u,N_v) = \lambda \mathrm{e}^{ \sum_{e\in\mathrm{path}(u,v)}\ln\alpha_{e}} = \lambda \mathrm{e}^{- \sum_{e\in\mathrm{path}(u,v)}\delta_{e}}
	= \lambda \mathrm{e}^{-D_{uv}}, \quad u,v\in\mathcal{V}.
	\label{eq:APSP}
\end{equation}
The APSP problem has been widely explored and a plethora of algorithms have been developed (see for instance \cite{chan2007more}). These algorithms could be used in our case; however, efficiently resolving the APSP problem was never a real issue on a tree, given their absence of cycles. Evidently, the shortest path between two vertices is the only path. We thus opt for the simple algorithm presented in \cite{harary1965structural}, Theorem~14.4, based on the powers of the cost matrix $\boldsymbol{c}$, with 
\begin{equation*}
	\boldsymbol{c}= (c_{ij},\; i\times j \in \mathcal{V}\times\mathcal{V}), \quad c_{ij}= \left\{\begin{array}{ll}
		\delta_{(i,j)}, &(i,j)\in\mathcal{E};\\
		0,& i=j;\\
		\infty, &\text{elsewhere}.
	\end{array}
	\right.
\end{equation*}
The algorithm utilizes the fact that $\boldsymbol{D}=\boldsymbol{c}^d$, with the powers of $\boldsymbol{c}$ according to $\star$, the min-plus matrix product defined as $(\boldsymbol{a}\star \boldsymbol{b})_{ij} = \min_{k=1}^d\{a_{ik}+b_{kj}\}$. One simply has to compute multiple min-plus matrix products to obtain the sought distance matrix. It is equivalent, and more well-suited in our case, to refer to the weighted adjacency matrix $\boldsymbol{A}$, given by 
\begin{equation}
	\boldsymbol{A}= (A_{ij},\; i\times j \in \mathcal{V}\times\mathcal{V}),\quad A_{ij} = \mathrm{e}^{-c_{ij}} = \left\{\begin{array}{ll}
		\alpha_{(i,j)}, &(i,j)\in\mathcal{E};\\
		1,&i=j;\\
		0, &\text{elsewhere},
	\end{array}
	\right.
 \label{eq:AdjacencyMatrix}
\end{equation}
with its powers according to $\sun$, which we define as the max-product matrix product: $(\boldsymbol{a}\;$\rotatebox[origin=c]{180}{$\sun$}$\;\boldsymbol{b})_{ij} = \max_{k=1}^d\{a_{ik}b_{kj}\}$. One indeed sees the equivalence:
\begin{align*}
	\mathrm{e}^{-\min_{k=1}^d\{\delta_{ik}+\delta_{kj}\}} &= \max_{k=1}^d\{\alpha_{ik}\alpha_{kj}\}.
\end{align*}
Given (\ref{eq:APSP}), we have $\mathrm{Cov}(N_u, N_v) = \lambda (\boldsymbol{A}^d)_{uv}$, with the powers according to $\sun$. This is implemented in Algorithm \ref{algo:VarCovarMatrix}.

\begin{algorithm}[H]
	\label{algo:VarCovarMatrix}
	\KwIn{Weighted adjacency matrix $\boldsymbol{A} = (A_{ij})_{i\times j\in \mathcal{V}\times\mathcal{V}}$; $\lambda$.}
	\KwOut{Variance-covariance matrix of $\boldsymbol{N}$.}
	 Set $\boldsymbol{A} = \boldsymbol{A}^{(1)}$. \\
	 \For{$k \in \{2,\ldots, d\}$}{
		 Compute $\boldsymbol{A}^{(k)} = \boldsymbol{A}^{(k-1)}\;$\rotatebox[origin=c]{180}{$\sun$}$\;\boldsymbol{A}$.\\
	}
	 Return $\lambda \boldsymbol{A}^{(d)}$.\\
	\caption{Computing the variance-covariance matrix of $\boldsymbol{N}$.}
\end{algorithm}

\section{Dependence ordering}	
\label{subsect:StochasticOrderingN}

To further our understanding of the dependence relations tying the components of MRFs from the family $\mathbb{MPMRF}$, we compare,
  using stochastic orders, their distributions according to the strength of their knit positive dependence. We rely on the supermodular order to make these comparisons.

\begin{deff}[Supermodular Order]
	Two vectors of random variables, $\boldsymbol{X}$ and $\boldsymbol{X}^{\prime}$, are ordered according to the supermodular order, denoted $\boldsymbol{X} \preceq_{sm} \boldsymbol{X}^{\prime}$, if 
	$\mathrm{E}[\varphi(\boldsymbol{X})] \leq \mathrm{E}[\varphi(\boldsymbol{X}^{\prime})]$
	for any supermodular function $\varphi$, given the expectations exist. A supermodular function $\varphi:\mathbb{R}^{d}\mapsto\mathbb{R}$ is such that 
	$\varphi(\boldsymbol{x}) + \varphi(\boldsymbol{x}^{\prime}) \leq  	\varphi(\boldsymbol{x}\wedge\boldsymbol{x}^{\prime}) + \varphi(\boldsymbol{x}\vee\boldsymbol{x}^{\prime})$, 
	holds for all $\boldsymbol{x},\boldsymbol{x}^{\prime}\in\mathbb{R}^d$, with $\wedge$ denoting the componentwise minimum and $\vee$ the componentwise maximum. 
\end{deff}

We recommend \cite{muller2002comparison} and \cite{shaked2007} for a good introduction to the supermodular order and to stochastic ordering in general. 
The following theorem explores the impact of stronger dependence parameters along the edges of underlying trees with the same shape.  

\begin{theorem}
	For a tree $\mathcal{T}=(\mathcal{V},\mathcal{E})$, consider $\boldsymbol{N} = (N_v, \, v \in \mathcal{V})\sim\text{MPMRF}(\lambda,\boldsymbol{\alpha}, \mathcal{T})$ and $\boldsymbol{N}^{\prime}=(N^{\prime}_v,\, v\in\mathcal{V})\sim\text{MPMRF}(\lambda,\boldsymbol{\alpha}^{\prime}, \mathcal{T})$. If, for every $e \in \mathcal{E}$, 
	$\alpha_e \leq \alpha_e^{\prime}$,        
	then $\boldsymbol{N} \preceq_{sm} \boldsymbol{N}^{\prime}$.
	\label{th:SupermodularOrderN}
\end{theorem}

\begin{proof}
	We first look at the case where only one dependence parameter differs from the distribution of $\boldsymbol{N}$ to that of $\boldsymbol{N}^{\prime}$, say respectively $\alpha_{(u,v)}$ and $\alpha^{\prime}_{(u,v)}$. Without loss of generality, assume the rooting of the underlying tree is such that $u=\mathrm{pa}(v) $ rather than vice-versa. We begin by considering only the two random variables associated to vertices $u$ and $v$. From Theorem \ref{th:JointPGF}, the joint probability generating function of $(N_{u},N_{v})$ is given by 
	\begin{equation*}
		\mathcal{P}_{N_{u},N_{v}}(t_u,t_v)= \mathrm{e}^{\lambda\left((1 - \alpha_{(u,v)})(t_{u}-1) + (1- \alpha_{(u,v)})(t_{v}-1) + \alpha_{(u,v)}(t_{u}t_{v} - 1)\right)}, \quad t_{u},t_{v} \in [-1,1],
	\end{equation*}
	which coincides with the joint probability generating function of a bivariate Poisson distribution with common shock (see \cite{teicher1954multivariate}); similarly for ($N_{u}^{\prime},N_{v}^{\prime}$). Therefore, as illustrated notably in Example 1 of \cite{meyer2015beyond}, since $\alpha_{(u,v)}\leq \alpha_{(u,v)}^{\prime}$,
	\begin{equation}
		(N_{u}, N_{v}) \preceq_{sm} (N_{u}^{\prime}, N_{v}^{\prime}).
		\label{eq:SupermodularN-Cond0-SMN1N2}
	\end{equation}
	
	Given the global Markov property of both $\boldsymbol{N}$ and $\boldsymbol{N}^{\prime}$, and since only $\alpha_{(u,v)}$ differs between $\boldsymbol{\alpha}$ and $\boldsymbol{\alpha}^{\prime}$, we have 
	\begin{equation}
		(N_{\mathrm{pa}(u)} | N_{u} = \theta) \stackrel{d}{=}( N_{\mathrm{pa}(u)}^{\prime} | N_{u}^{\prime} = \theta),\quad \theta\in\mathbb{N}.
		\label{eq:SupermodularN-distr1}
	\end{equation}
	Similarly, for every ${j} \in \mathrm{ch}(v)$,
	\begin{equation}
		(N_{j} | N_{v} = \theta) \stackrel{d}{=} (N_{j}^{\prime} | N_{v}^{\prime} = \theta), \quad \theta\in\mathbb{N}.
		\label{eq:SupermodularN-distr2}
	\end{equation} 
	Besides, from Theorem~ \hyperref[th:PropertyBinThinOp]{\ref{th:PropertyBinThinOp}(g)}, the order relation
	\begin{equation}
		(N_{w} | N_{\mathrm{pa}(w)} = \theta_1) \preceq_{st} (N_{w} | N_{\mathrm{pa}(w)} = \theta_2), \quad \theta_1, \theta_2 \in \mathbb{N}, \; \theta_1 \leq \theta_2,
		\label{eq:SupermodularN-Cond1-StoDomN}
	\end{equation}
	holds for all $w \in \mathcal{V}\backslash\{r\}$,
	where $\preceq_{st}$ denotes the first order stochastic dominance. 
	By Theorem \ref{th:Reversibility}, relation 
	\begin{equation}
		(N_{\mathrm{pa}(w)} | N_{w} = \theta_1) \preceq_{st} (N_{\mathrm{pa}(w)} | N_{w} = \theta_2), \quad \theta_1, \theta_2 \in \mathbb{N}, \; \theta_1 \leq \theta_2,
		\label{eq:SupermodularN-Cond2-StoDomN}
	\end{equation}
	is also true for all $w \in \mathcal{V}\backslash\{r\}$. 
	Similarly, relations in (\ref{eq:SupermodularN-Cond1-StoDomN}) and (\ref{eq:SupermodularN-Cond2-StoDomN}) also hold for $N_w^{\prime}$ for all $w \in \mathcal{V}\backslash\{r\}$. Given the global Markov properties of $\boldsymbol{N}$ and $\boldsymbol{N}^{\prime}$, and relations (\ref{eq:SupermodularN-Cond0-SMN1N2}),
	(\ref{eq:SupermodularN-distr1}),
	(\ref{eq:SupermodularN-distr2}),
	(\ref{eq:SupermodularN-Cond1-StoDomN}) and (\ref{eq:SupermodularN-Cond2-StoDomN}), we can invoke Theorem 9.A.15 of \cite{shaked2007} by using $N_{u}$ and $N_{v}$ themselves as mixing random variables to deduce
	\begin{equation*}
		(N_{u}, N_{v}, N_{\mathrm{pa}(u)}, (N_{j}: j \in \mathrm{ch}(v))) \preceq_{sm} (N_{u}^{\prime}, N_{v}^{\prime}, N_{\mathrm{pa}(u)}^{\prime}, (N_{j}^{\prime}: j \in \mathrm{ch}(v))).
	\end{equation*}
	We again refer to Theorem 9.A.15 of \cite{shaked2007} and use $(N_{u}, N_{v}, N_{\mathrm{pa}(u)}, (N_{j}: j \in \mathrm{ch}(v)))$ as mixing random variables, and repeat until all vertices have been captured to conclude that 
	$\boldsymbol{N} \preceq_{sm} \boldsymbol{N}^{\prime}$.
	Finally, from the transitivity of the supermodular order, the result follows by repeating this rationale, successively substituting each $\alpha_e$ for $\alpha_e^{\prime}$ for all $e\in \mathcal{E}$. 
\end{proof}

In the following remark, we invert the paradigm of Theorem \ref{th:SupermodularOrderN}, meaning we modify the shape of the underlying trees but keep the dependence parameters unchanged. 

\begin{rem}
	Consider trees $\mathcal{T}=(\mathcal{V},\mathcal{E})$ and $\mathcal{T}^{\prime}=(\mathcal{V}, \mathcal{E}^{\prime})$. Let $\boldsymbol{N} = (N_v, \, v \in \mathcal{V})\sim\text{MPMRF}(\lambda,\boldsymbol{\alpha}, \mathcal{T})$ and $\boldsymbol{N}^{\prime}=(N^{\prime}_v,\, v\in\mathcal{V})\sim\text{MPMRF}(\lambda,\boldsymbol{\alpha}^{\prime}, \mathcal{T}^{\prime})$, where for every $e\in\mathcal{E}$, there is a corresponding $e^{\prime}\in\mathcal{E}^{\prime}$ such that $\alpha_e=\alpha^{\prime}_{e^{\prime}}$. 
	If $\mathcal{E} \neq \mathcal{E}^{\prime}$, meaning if trees have different shapes, then $\boldsymbol{N}$ and $\boldsymbol{N}^{\prime}$ are necessarily not comparable in the sense of the supermodular order. Both trees $\mathcal{T}$ and $\mathcal{T}^{\prime}$ have the same number of vertices and edges, but different shapes. This implies that there exist $u,v,w \in \mathcal{V}$ such that $(u,v) \in \mathcal{E}$, $(u, v) \not\in \mathcal{E}^{\prime}$, $(u,w) \not\in \mathcal{E}$ and $(u,w) \in \mathcal{E}^{\prime}$. Since the dependence parameters are the same across both trees, we necessarily have 
	$
	\mathrm{Cov}(N_{u}, N_{v}) \geq \mathrm{Cov}(N_{u}^{\prime}, N_{v}^{\prime})
	$
	and 
	$
	\mathrm{Cov}(N_{u}, N_{w}) \leq \mathrm{Cov}(N_{u}^{\prime}, N_{w}^{\prime})
	$
	by Theorem \ref{th:Covariance}, preventing to establish a supermodular ordering between $\boldsymbol{N}$ and $\boldsymbol{N}^{\prime}$. 
	Indeed, this would contravene Theorem 3.9.5(c) of \cite{muller2002comparison}, stating that for two random vectors $\boldsymbol{X}$ and $\boldsymbol{X}^{\prime}$, such that $\boldsymbol{X} \preceq_{sm} \boldsymbol{X}^{\prime}$, the covariance between each pair of components of $\boldsymbol{X}$ must be inferior to that of the corresponding components of $\boldsymbol{X}^{\prime}$.
	\label{th:NoSupermodularNShape}
\end{rem}

To illustrate Remark~\ref{th:NoSupermodularNShape}, we provide the following example. 
\begin{ex} 
	\label{ex:NoSupermodularNShape}	
	Consider $\boldsymbol{N}$ and $\boldsymbol{N}^{\prime}$ respectively defined on trees $\mathcal{T}$ and $\mathcal{T}^{\prime}$ of Fig.~\ref{fig:TwoTrees}. Suppose both have the same parameter $\lambda$ for their marginal distributions, and the same dependence parameters $\alpha_e$ on their common edges. Suppose $\alpha_{(2,5)} = \alpha_{(3,5)}$. Trees $\mathcal{T}$ and $\mathcal{T}^{\prime}$ only differ by the position of vertex 5. Applying Theorem~\ref{th:Covariance}, we note
	\begin{align*}
		&\mathrm{Cov}(N_2,N_5) = \lambda\alpha_{(2,5)} \geq \lambda\alpha_{(3,5)} \alpha_{(1,3)} \alpha_{(1,2)} = \mathrm{Cov}(N_2^{\prime},N_5^{\prime});\\
		&\mathrm{Cov}(N_3,N_5) = \lambda\alpha_{(2,5)} \alpha_{(1,2)} \alpha_{(1,3)} \leq \lambda\alpha_{(3,5)} = \mathrm{Cov}(N_3^{\prime},N_5^{\prime}).
	\end{align*} 
	Therefore, $\boldsymbol{N}$ and $\boldsymbol{N}^{\prime}$ are not comparable under the supermodular order. 
\end{ex}

\begin{figure}[H]
	\centering
	\begin{tikzpicture}[every node/.style={text=Black, circle, draw = Maroon, inner sep=0mm, outer sep = 0mm, minimum size=3.5mm, fill = White}, node distance = 1.5mm, scale=0.1, thick]
		\node (1a) {\tiny $1$};
		\node [below left=of 1a] (2a) {\tiny $2$};
		\node [below right=of 1a] (3a) {\tiny $3$};
		\node [below left =of 2a] (4a) {\tiny $4$};
		\node [draw =Sepia, below =of 2a] (5a) {\tiny $5$};
		\node [below =of 3a] (6a) {\tiny $6$};
		\node [below right =of 3a] (7a) {\tiny $7$};
		
		\draw (1a) -- (2a);
		\draw (1a) -- (3a);
		\draw (2a) -- (4a);
		\draw (2a) -- (5a);            
		\draw (3a) -- (6a);       
		\draw (3a) -- (7a);

		\node [right = 27mm of 1a](1b) {\tiny $1$};
		\node [below left=of 1b] (2b) {\tiny $2$};
		\node [below right=of 1b] (3b) {\tiny $3$};
		\node [below left =of 2b] (4b) {\tiny $4$};
		
		\node [below =of 3b] (6b) {\tiny $6$};
		\node [below right =of 3b] (7b) {\tiny $7$};
		\node [draw =Sepia,  left = 0.5mm of 6b] (5b) {\tiny $5$};
		
		\draw (1b) -- (2b);
		\draw (1b) -- (3b);
		\draw (2b) -- (4b);
		\draw (3b) -- (5b);            
		\draw (3b) -- (6b);       
		\draw (3b) -- (7b);

		\node[draw=none, below = 10mm of 1a](t1){$\mathcal{T}^{\mathcolor{White}{\prime}}$};
		\node[draw=none, below= 10mm of 1b](t2){$\mathcal{T}^{\prime}$};
	\end{tikzpicture}
	\caption{Trees $\mathcal{T}$ and $\mathcal{T}^{\prime}$ of Example  \ref{ex:NoSupermodularNShape}.}
	\label{fig:TwoTrees}
\end{figure}

\section{Sampling}
\label{sect:Simulation}

The stochastic representation of $\boldsymbol{N}$ given by (\ref{eq:StoDynamics}) in Theorem~\ref{th:StoDynamics} comes with a convenient sampling procedure that scales well to high dimensions $d$, as it allows to produce the realizations for every component of $\boldsymbol{N}$ successively. The order of conditioning picked to derive the expression of the joint probability mass function in Theorem~\ref{th:JointPMF} -- having the parent vertex precede its children, given a chosen root $r\in\mathcal{V}$ -- is the same order in which one will produce realizations of $N_v$, $v\in\mathcal{V}$. Obviously, given Theorem~\ref{th:Reversibility}, any rooting may be chosen to define the filial relations, and thus the order of sampling. The gist of the sampling method is then straightforward: provided a realization of $N_{\mathrm{pa}(v)}$, we independently produce a realization of a binomial random variable for the propagation part of (\ref{eq:StoDynamics}) and of a Poisson random variable for the innovation part and sum both realizations to produce that of $N_v$.

We translate the discussed method to Algorithm~\ref{algo:simulNi}. The weighted adjacency matrix, as defined in (\ref{eq:AdjacencyMatrix}), provided in input needs to be constructed according to a topological order given a root $r\in\mathcal{V}$, meaning the row and column index associated to a vertex must be greater than that of its parent. Consequently, first row and column are associated to the root $r$.  The algorithm then proceeds sequentially along the row indices; this naturally supplies it with the order of sampling discussed above. Moreover, because the matrix is constructed in topological order, for a certain row $j$, $j\in\{2,\ldots,d\}$, the first non-zero element has column index $\mathrm{pa}(j)$. This feature is employed to easily retrieve the filial relations in Algorithm~\ref{algo:simulNi}. Example~\ref{ex:topologicalorder} shows two weighted adjacency matrices in topological order.

\begin{algorithm}[H]
	\label{algo:simulNi}
	\caption{Stochastic representation sampling method.}
	\KwIn{Weighted adjacency matrix $\boldsymbol{A} = (A_{ij})_{i\times j\in \mathcal{V}\times\mathcal{V}}$; $\lambda$.}
	\KwOut{Realization of $\boldsymbol{N}$.} 
	 Sample $N_1\sim\mathrm{Poisson}(\lambda)$. \\
	 \For{$k=2,\ldots,d$}{
		 Compute $\pi_k = \inf\{j:A_{kj}>0\}$.\\
		 Sample $B_k\sim\mathrm{Binomial}(N_{\pi_k}, A_{\pi_k k})$ random variable.\\
		 Sample $L_k\sim\mathrm{Poisson}(\lambda(1-a_k))$ random variable. \\
		 Compute $N_k = B_k + L_k$.\\}
	 Return $\{N_i,\;i\in\{1,\ldots,d\}\}$.
\end{algorithm}

\begin{ex}
\label{ex:topologicalorder}
	Consider the 7-vertex tree of Illustration~\ref{fig:ExTreeNotations}. Weighted adjacency matrices constructed in a topological order according to $\mathcal{T}_1$ and $\mathcal{T}_3$, respectively, are 
	\begin{equation*}
		\boldsymbol{A}_{\mathcal{T}_1} = 	\begin{bmatrix}
			1&{\alpha_{(1,2)}}&{\alpha_{(1,3)}}&0&0&0&0\\
			{\alpha_{(1,2)}}&1&0&0&0&0&0\\
			{\alpha_{(1,3)}}&0&1&{\alpha_{(3,4)}}&{\alpha_{(3,5)}}&0&0\\
			0&0&{\alpha_{(3,4)}}&1&0&{\alpha_{(4,6)}}&{\alpha_{(4,7)}}\\
			0&0&{\alpha_{(3,5)}}&0&1&0&0\\
			0&0&0&{\alpha_{(4,6)}}&0&1&0\\
			0&0&0&{\alpha_{(4,7)}}&0&0&1
		\end{bmatrix}; \;\boldsymbol{A}_{\mathcal{T}_3} = 	\begin{bmatrix}
		1&{\alpha_{(1,3)}}&0&{\alpha_{(3,4)}}&0&0&{\alpha_{(3,5)}}\\
		{\alpha_{(1,3)}}&1&{\alpha_{(1,2)}}&0&0&0&0\\
		0&{\alpha_{(1,2)}}&0&1&0&0&0\\
		{\alpha_{(3,4)}}&0&0&1&{\alpha_{(4,6)}}&{\alpha_{(4,7)}}&0\\
		0&0&0&{\alpha_{(4,6)}}&1&0&0\\
		0&0&0&{\alpha_{(4,7)}}&0&1&0\\
		{\alpha_{(3,5)}}&0&0&0&0&0&1
	\end{bmatrix},
	\end{equation*}
where in the first matrix the columns and rows are labelled in order (1,2,3,4,5,6,7) and in the second matrix, (3,1,2,4,6,7,5).  
\end{ex}

Algorithm~\ref{algo:simulNi} requires $2d-1$ random numbers to produce a realization of the $d$-variate random vector $\boldsymbol{N}$. For a high dimension $d$ $(d=100,1000,\ldots)$, 
one may generate thousands of sample points in matters of seconds with a regular personal computer and thus appraise the efficiency of the presented method.   
The arising of a straightforward sampling method for MRFs having a clique-based stochastic representation, hence explicit expressions for the probability mass functions involved in an iterative conditioning as in (\ref{eq:MarkovProba}), is notably remarked in \cite{pickard1977curious} and \cite{pickard1980unilateral}. As discussed in Section~\ref{subsect:PMF-PGF}, most families of MRFs do not exhibit this feature 
and cannot produce the iterative sequence of simulation on which Algorithm~\ref{algo:simulNi} relies.
Chapter~12 of \cite{koller2009probabilistic} presents some of the methods usually employed to sample from MRFs, namely the Gibbs sampler and importance sampling. The Gibbs sampler is an implementation of Markov chain Monte Carlo to generate a sample from a distribution when only conditional probabilities are known. Following this method, one realization of the MRF requires simulating across the whole tree multiple times, iteratively refining the simulation; and hence, requires more computing ressources than a straightforward simulation. For further information on the Gibbs sampler and its derivatives in the context of discrete MRFs, see \cite{izenman2021sampling}. Importance sampling, on the other hand, may only require $d$ random numbers. The idea is to generate from another, similar, multivariate distribution and weight the importance of the samples according to a ratio of its joint probability mass function and that of the desired MRF. While importance sampling is very efficient, having a direct sampling procedure clears the need of finding a distribution similar enough.
Monte Carlo estimations and computations are also subject to more variance when employing importance sampling; this variance is governed by the degree of similarity of the proxy distribution (refer to equation (12.14) of \cite{koller2009probabilistic}). One may consult \cite{tokdar2010importance} for a review on importance sampling.

\section{Sum of the components of \textbf{\textit{N}}}
\label{sect:Sum}
Let us denote by $M$ the sum of the components of $\boldsymbol{N}$,  that is, $M = \sum_{v \in \mathcal{V}} N_v$. Analyzing the behavior of $M$ is of high interest to understand the aggregate dynamics of systems described by $\boldsymbol{N}$. We derive in this section the distribution of $M$ and study the impact of the dependence structure embedded in $\boldsymbol{N}$ on the sum $M$ using stochastic ordering. We then quantify the contribution of each component of $\boldsymbol{N}$ to $M$ through expected allocations.

\subsection{Distribution of the sum}

The analytic expression for the joint probability generating function of $\boldsymbol{N}$, provided in Theorem~\ref{th:JointPGF}, also grants information on the distribution of $M$, as it leads to an expression for the probability generating function of $M$ through the relation 
\begin{equation}
	\mathcal{P}_M(t) =  \mathcal{P}_{\boldsymbol{N}}(t,\ldots,t), \quad t \in [-1,1].
	\label{eq:PGFofMfromJointPGFofN}
\end{equation}
We rely on (\ref{eq:PGFofMfromJointPGFofN}) for two purposes: first, to derive general results on the distribution of $M$, provided in Theorem~\ref{th:DistrM} below; second, to compute the probability mass function of $M$ by directly programming the joint probability generating function of $\boldsymbol{N}$ and using the Fast Fourier Transform (FFT) algorithm, as indicated in Algorithm~\ref{algo:fftM}. For notation purposes, let $\boldsymbol{1}_k$ denote a $k$-long vector of ones, $k\in\mathbb{N}^*$, such that, for instance, (\ref{eq:PGFofMfromJointPGFofN}) is rewritten as $\mathcal{P}_M(t) =  \mathcal{P}_{\boldsymbol{N}}(t\,\boldsymbol{1}_{d})$, $t \in [-1,1]$. The following theorem shows that $M$ follows a compound Poisson distribution.

\begin{theorem}[Distribution of the sum]
	\label{th:DistrM}
	For a tree $\mathcal{T}=(\mathcal{V},\mathcal{E})$, consider $\boldsymbol{N} = (N_v, \, v \in \mathcal{V})\sim\text{MPMRF}(\lambda,\boldsymbol{\alpha}, \mathcal{T})$. The random variable $M$, corresponding to the sum of the components of $\boldsymbol{N}$, has a compound Poisson distribution with primary distribution parameter $\lambda_{M} = \lambda\left(d-\sum_{e\in\mathcal{E}}\alpha_{e}\right)$ and probability generating function of the secondary distribution $\mathcal{P}_{C_M}(t) = \sum_{v\in\mathcal{V}}\frac{1-\alpha_{(\mathrm{pa}(v),v)}}{d-\sum_{e\in\mathcal{E}}\alpha_{e}}\eta_{v}^{\mathcal{T}_r}(t\vecun{v})$, yielding $\mathrm{E}[C_M] = \frac{d}{d - \sum_{e\in\mathcal{E}}\alpha_{e}}$.
\end{theorem}
\begin{proof}
	The probability generating function of a random variable $X$ following a compound Poisson distribution is given by
	\begin{equation}
		\mathcal{P}_{X}(t) = \mathcal{P}_{\Lambda_X}(\mathcal{P}_{C_X}(t)) = \mathrm{e}^{\lambda_X(\mathcal{P}_{C_X}(t)-1)}, \quad \lambda>0,\,t\in[-1,1],
		\label{eq:PGFCompoundPoisson}
	\end{equation}
	where $\mathcal{P}_{\Lambda_X}$ is the probability generating function of the primary (or counting) distribution of $X$ and $\mathcal{P}_{C_X}$ the probability generating function of its secondary distribution. Given (\ref{eq:PGFofMfromJointPGFofN}), we derive the probability generating function of $M$ from (\ref{eq:JointPGF}) as follows
	\begin{align}
		\mathcal{P}_{M}(t) &= \prod_{v\in\mathcal{V}}\mathrm{e}^{\lambda(1-\alpha_{(\mathrm{pa}(v),v)})(\eta_{v}^{\mathcal{T}_r}(t\vecun{v}) - 1)}
  = \mathrm{e}^{\sum_{v\in\mathcal{V}}\lambda(1-\alpha_{(\mathrm{pa}(v),v)})(\eta_{v}^{\mathcal{T}_r}(t\vecun{v}) - 1)}
  = \mathrm{e}^{\lambda\left(d-\sum_{e\in\mathcal{E}}\alpha_{e}\right) \left( \left(\sum_{v\in\mathcal{V}}\frac{1-\alpha_{(\mathrm{pa}(v),v)}}{d-\sum_{e\in\mathcal{E}}\alpha_{e}}\eta_{v}^{\mathcal{T}_r}(t\vecun{v})\right)-1\right)}.\label{eq:proofPGFofM-1}
	\end{align}
	Comparing (\ref{eq:PGFCompoundPoisson}) and (\ref{eq:proofPGFofM-1}), we deduce that $M$ follows a compound Poisson distribution with $\lambda_M =  \lambda(d-\sum_{e\in\mathcal{E}}\alpha_{e})$ and $\mathcal{P}_{C_M}(t) = \sum_{v\in\mathcal{V}}\frac{1-\alpha_{(\mathrm{pa}(v),v)}}{d-\sum_{e\in\mathcal{E}}\alpha_{e}}\eta_{v}^{\mathcal{T}_r}(t\vecun{v})$. Since $\eta_{v}^{\mathcal{T}_r}$, $v\in\mathcal{V}$, is a joint probability generating function, as discussed in Section \ref{subsect:PMF-PGF}, $\mathcal{P}_{C_M}$ is a probability generating function as well, given (\ref{eq:PGFofMfromJointPGFofN}). Finally, since $\mathrm{E}[M] = \lambda d$, and since $\mathrm{E}[M] = \lambda_M\mathrm{E}[C_M]$ by Wald's Lemma, we find $\mathrm{E}[C_M] = \frac{d}{d - \sum_{e\in\mathcal{E}}\alpha_{e}}$. 
\end{proof}

From Theorem~\ref{th:DistrM}, one notices that $\lambda_M$ and $\mathrm{E}[C_M]$ do not depend on the shape of $\mathcal{T}_r$. This arises from the fact that $\mathrm{Pr}(M=0)$ does not depend on the shape of the tree either, since
\begin{equation*}
	\mathrm{Pr}(M=0) = \mathrm{Pr}\left(\bigcap_{v\in\mathcal{V}}\{N_v = 0\}\right)\notag\\
	= \mathrm{Pr}\left(\bigcap_{v\in\mathcal{V}}\{L_v = 0\}\right)\notag\\
	= \prod_{v\in\mathcal{V}}\mathrm{Pr}(L_v=0)
\end{equation*}
due to the independence assumption between components within $\boldsymbol{L}$. With $L_v$ following a Poisson distribution of parameter $\lambda(1-\alpha_{(\mathrm{pa}(v),v)})$, $v\in\mathcal{V}$, we obtain
\begin{equation*}
	\mathrm{Pr}(M=0) 
	= \prod_{v\in\mathcal{V}} \mathrm{e}^{-\lambda(1-\alpha_{(\mathrm{pa}(v),v)})}\notag\\ 
	= \mathrm{e}^{-\lambda(d-\sum_{v\in\mathcal{V}}\alpha_{(\mathrm{pa}(v),v)})},
\end{equation*}
which also highlights the value of $\lambda_M$. Thus, the shape of $\mathcal{T}_r$ is entirely expressed through the distribution of $C_M$, without affecting $\mathrm{E}[C_M]$. This is explored in the following corollary.

\begin{cor}[Secondary distributions]
	\label{th:SeverityDistrofM}
	For a tree $\mathcal{T}=(\mathcal{V},\mathcal{E})$, consider $\boldsymbol{N} = (N_v, \, v \in \mathcal{V})\sim\text{MPMRF}(\lambda,\boldsymbol{\alpha}, \mathcal{T})$. Let $\alpha_e = \alpha$ for every $e\in\mathcal{E}$, where $\alpha\in[0,1]$. If the underlying tree $\mathcal{T}$ to $\boldsymbol{N}$ is 
	\begin{itemize} [nosep]
		\item a $d$-vertex star, then the distribution of ${C_M}$ is a mixture of a degenerate distribution at 1 and a binomial distribution with size parameter $d-1$ and probability parameter $\alpha$, and shifted by one; 
		\item a $d$-vertex series tree, then the distribution of ${C_M}$ is a mixture of  geometric distributions of parameter $\alpha$, shifted by one and censored at $i$, with $i$ going from 1 to $d$; 
		\item a $\chi$-nary tree of radius $\xi$, then the distribution of ${C_M}$ is a mixture of distributions representing the sum up to the $i$th generation of a Galton-Watson process whose offspring distribution is a binomial distribution with size parameter $\chi$ and probability parameter $\alpha$, with $i$ going from 0 to $\xi$. 
	\end{itemize} 
	
\end{cor}

\begin{proof}
	We derive the probability generating function of $M$ under each structure by relying on (\ref{eq:proofPGFofM-1}). Under the assumption that $\alpha_e = \alpha$ for every $e\in\mathcal{E}$, $\alpha\in[0,1]$, the probability generating function of $M$ encrypted on a $d$-vertex star-shaped tree is given by
	\begin{equation*}
		\mathcal{P}_M(t) = \mathrm{e}^{\lambda(\alpha+(1-\alpha)d)\left(\frac{1}{\alpha+(1-\alpha)d}t(1-\alpha+\alpha t)^{d-1} + \frac{(1-\alpha)(d-1)}{\alpha+ (1-\alpha)d}t - 1\right)}, \quad d\geq 2,\; \lambda>0;
	\end{equation*}
	the probability generating function of $M$ encrypted on a $d$-vertex series-shaped tree is given by
	\begin{equation*}
		\mathcal{P}_M(t) = \mathrm{e}^{\lambda (\alpha+(1-\alpha)d) \left( \sum_{i=1}^{d}\frac{(1-\alpha)^{\min(1,d-i)}}{\alpha+(1-\alpha)d} \sum_{j=1}^{i} (1-\alpha)^{\min(1,i-j)}\alpha^{j-1}t^j  - 1\right)}, \quad d\geq 2,\; \lambda>0; 
	\end{equation*}
	and the probability generating function of $M$ encrypted on a $\chi$-nary tree of radius $\xi$ is given by
	\begin{equation*}
		\mathcal{P}_{M}(t) = \mathrm{e}^{\lambda \left(\alpha+ (1-\alpha)\frac{\chi^{\xi+1}-1}{\chi-1}\right) \left( \sum_{i=0}^{\xi} \frac{\chi^{\xi-i}(1-\alpha)^{\min(1,\xi-i)}}{\alpha+ (1-\alpha)\frac{\chi^{\xi+1}-1}{\chi-1}} \psi^{\{i\}}(t, t, \alpha, \chi)  - 1\right)} \quad \chi\geq1,\; \xi \geq 1,\; \lambda>0, 
	\end{equation*}
	where 
	\begin{equation}
		\psi(y, t,\alpha,\chi) = t(1 - \alpha + \alpha y)^{\chi},
		\label{eq:psi}
	\end{equation}
	with $\psi^{\{k\}}$ denoting the $k$th composition of $\psi$ on $y$, $k\in \mathbb{N}^*$, and $\psi^{\{0\}}(y, t,\alpha,\chi)=t$. The results then follow by identification of the probability generating functions, uniquely describing the distributions. To effectively see that $\psi^{\{i\}}(t,t,\alpha,\chi)$ is the probability generating function of the sum of generations 0 to $i$ of a Galton-Watson process, one may refer to Section 1.13.2 of \cite{harris1963theory}.  
\end{proof}

Relying on (\ref{eq:PGFofMfromJointPGFofN}), Algorithm~\ref{algo:fftM} exploits the efficiency of the FFT algorithm and the recursiveness of the joint probability generating function of $\boldsymbol{N}$ to compute the probability mass function of $M$ in a reasonable computer time, even for large dimensions $d$. As for Algorithm~\ref{algo:simulNi}, the weighted adjacency matrix provided in input must be constructed in a topological order given a chosen root $r\in\mathcal{V}$; Algorithm~\ref{algo:fftM} utilizes this construction to recursively compute the values of the probability generating functions $\{\eta_v^{\mathcal{T}_r}(t\vecun{v}), \, v\in\mathcal{V}\}$. In Section~\ref{sect:NumericalExamples}, we provide an example of application of this algorithm.

\begin{algorithm}[H]
	\label{algo:fftM}
	\caption{Computing the probability mass function of $M$.} 
	\KwIn{Weighted adjacency matrix $\boldsymbol{A} = (A_{ij})_{i\times j\in \mathcal{V}\times \mathcal{V}}$; $\lambda$.}
	\KwOut{Vector $\boldsymbol{p}^{(M)}=(p^{(M)}_k)_{k\in\{1,\ldots,k_{\max}\}}$ such that $p^{(M)}_k$ = $p_{M}(k-1)$.}
	 Set $n_{\mathrm{fft}}$ to be a large power of 2 (e.g. $2^{15}$). This determines $k_{\max}$. \\
	 Set $\boldsymbol{b} = (b_{i})_{i\in\{1,\ldots,n_{\mathrm{fft}}\}} = (0,1,0,0,\ldots,0)$.\\
	 Use fft to compute the discrete Fourier transform $\boldsymbol{\phi}^{(b)}$ of $\boldsymbol{b}$. \\
	 \For{$\ell = 1,\ldots, n_{\mathrm{fft}}$}{
		 Set $\boldsymbol{H} = (H_{ij})_{i\times j \in \mathcal{V}\times\mathcal{V}}$ to be an all-1 matrix. \\
		 \For{$k = d,d-1,\ldots,2$} {
			 Compute $\pi_k = \inf\{j:A_{kj} >0\}$.\\
			 Compute $h_k = {\phi}_\ell^{(b)} \prod_j H_{kj}$.\\
			 Overwrite $H_{\pi_kk}$ to be $(1-A_{\pi_kk})+A_{\pi_kk}h_k$.\\}
		 Compute $h_1 = {\phi}_\ell^{(b)} \prod_j H_{1j}$.\\
		 Compute \begin{equation}\phi_\ell^{(M)} = \prod_k \exp(\lambda (1-{A}_{\pi_kk})(h_k-1)).\label{eq:algo}\end{equation}\\}
	 Use fft to compute the inverse discrete Fourier transform $\boldsymbol{p}^{(M)}$ of $\boldsymbol{\phi}^{(M)}$.\\
	 Return $\boldsymbol{p}^{(M)}$.\\
\end{algorithm}

\subsection{Stochastic ordering of the sum and impact of dependence}
\label{subsect:StochacticOrderingM}

We use the convex order to explore the impact of an increase of $\boldsymbol{\alpha}$ on $M$ given a fixed underlying tree. We first recall the definition of the convex order. 
\begin{deff}[Convex order]
	Two random variables $X$ and $X^{\prime}$ are said to be ordered according to the convex order, denoted $X\preceq_{cx}X^{\prime}$, if 
	$ \mathrm{E}[\varphi(X)] \leq \mathrm{E}[\varphi(X^{\prime})]$ 
	for any convex function $\varphi$, given the expectations exist. 
\end{deff}

For an in-depth look at the convex order, one may consult Chapters 1 and 2 of \cite{muller2002comparison} and Chapter 3 of \cite{shaked2007}.

\begin{cor}
	For a tree $\mathcal{T}=(\mathcal{V},\mathcal{E})$, consider $\boldsymbol{N} = (N_v, \, v \in \mathcal{V})\sim\text{MPMRF}(\lambda,\boldsymbol{\alpha}, \mathcal{T})$ and $\boldsymbol{N}^{\prime}=(N^{\prime}_v,\, v\in\mathcal{V})\sim\text{MPMRF}(\lambda,\boldsymbol{\alpha}^{\prime}, \mathcal{T})$. Let $M$ and $M^{\prime}$ be the sums of the components of $\boldsymbol{N}$ and $\boldsymbol{N}^{\prime}$, respectively. If, for every $e \in \mathcal{E}$, 
	$\alpha_e \leq \alpha_e^{\prime}$, 
	then $M \preceq_{cx} M^{\prime}$.
	\label{th:ConvexOrderM}
\end{cor}

\begin{proof}
	The result follows from Theorem \ref{th:SupermodularOrderN} and Theorem 3.1 of \cite{muller1997stop}. 
\end{proof}

The convex order should be interpreted as such: although they are on the same scale, $M^{\prime}$ encapsulates more variability than $M$. In this sense, an implication of the convex order is that, if $M\preceq_{cx}M^{\prime}$, then $\mathrm{E}[M]=\mathrm{E}[M^{\prime}]$ and $\mathrm{Var}(M)\leq \mathrm{Var}(M^{\prime})$. The results of Corollary~\ref{th:ConvexOrderM} thus shows that stronger dependence induces more variability for the sum $M$, which is intuitive.

\subsection{Expected allocations}
\label{sect:ExpAlloc}

A quantity that allows to understand the subsumed behavior of each constituent of $M$ in an aggregate dynamic is the expected allocation of an element of $\boldsymbol{N}$ to $M$, defined in our context as 
$\mathrm{E}[N_v\mathbb{1}_{\{M=k\}}]$, $k\in\mathbb{N}$, $v\in\mathcal{V}$.
Expected allocations provide a way to dive into the underlying dynamics of $\boldsymbol{N}$. In insurance risk modeling for example, expected allocations may be used to compute conditional mean risk-sharing, introduced in \cite{denuit2012convex}, given by
\begin{equation*}
	\mathrm{E}[N_v|M=k] = \frac{\mathrm{E}[N_v\mathbb{1}_{\{M=k\}}]}{\mathrm{Pr}(M=k)},\quad k\in\mathbb{N},\; v\in\mathcal{V}.
    \label{eq:cmrs}
\end{equation*}
For analytical results on the conditional mean risk-sharing rule within the context of probabilistic graphical models, see \cite{denuit2022conditional}. Let us define the Tail-Value-at-Risk (TVaR) at confidence level $\kappa$, for a random variable $X$, as $\mathrm{TVaR}_{\kappa}(X)=\tfrac{1}{1-\kappa}\int_\kappa^{1}\mathrm{VaR}_u(X)\mathrm{d}u$, where $\mathrm{VaR}_{\kappa}(X) = \inf\{x\in\mathbb{R}:\,F_X(x) \geq \kappa\}$, $\kappa\in[0,1)$. Another noteworthy application of expected allocations is for the computation of contributions to the TVaR under Euler's principle. In capital allocation and risk theory, contributions to the TVaR serve to quantify the portion of the aggregate risk enclosed within each component of a random vector; see \cite{tasche2007capital}. 
Since $M$ takes on values in $\mathbb{N}$, the contribution of $N_v$ to the TVaR, for each $v\in\mathcal{V}$, is given by 
 \begin{align}
 	\mathcal{C}^{\mathrm{TVaR}}_{\kappa}(N_v;\, M) &= \frac{1}{1-\kappa}\left( \mathrm{E}[N_v\mathbb{1}_{\{M>\mathrm{VaR}_{\kappa}(M)\}}] + \frac{F_M(\mathrm{VaR}_{\kappa}(M))-\kappa}{p_M(\mathrm{VaR}_{\kappa}(M))} \mathrm{E}[N_v\mathbb{1}_{\{M=\mathrm{VaR}_{\kappa}(M)\}}]\right) \notag\\
  &= \frac{1}{1-\kappa}\left( \mathrm{E}[N_v] -  \sum_{i=0}^{\mathrm{VaR}_{\kappa}(M)}\mathrm{E}[N_v\mathbb{1}_{\{M=i\}}] + \frac{F_M(\mathrm{VaR}_{\kappa}(M))-\kappa}{p_M(\mathrm{VaR}_{\kappa}(M))} \mathrm{E}[N_v\mathbb{1}_{\{M=\mathrm{VaR}_{\kappa}(M)\}}]\right),\label{eq:contribTVaR}
\end{align}
 for $\kappa \in [0,1)$; see, for instance, Section~2 in \cite{mausser2018long}. Note that $\sum_{v\in\mathcal{V}} \mathcal{C}^{\mathrm{TVaR}}_{\kappa}(N_v;\, M) = \mathrm{TVaR}_{\kappa}(M)$ holds for $\kappa \in [0,1)$.
In \cite{blier2022generating}, the authors discuss the multiple applications of expected allocations and introduce the ordinary generating function associated to the sequence $\{\mathrm{E}[N_v\mathbb{1}_{\{M=k\}}],\, k\in\mathbb{N}\}$, for the purpose of efficient computations.

\begin{deff}[OGFEA]
	\label{def:OGFEA}
	Let $\boldsymbol{X} = (X_v,\, v\in\mathcal{V})$ be a vector of random variables taking values in $\mathbb{N}^d$ and $Y$ be the sum of its components.
	The ordinary generating function of expected allocations (OGFEA) of $X_v$ to $Y$, $v\in\mathcal{V}$, is $\mathcal{P}_Y^{[v]}$ such that
	\begin{equation*}
		\mathcal{P}_Y^{[v]}(t) = \sum_{k=0}^{\infty}t^k\mathrm{E}[X_v\mathbb{1}_{\{Y=k\}}], \quad t\in[-1, 1],\,v\in\mathcal{V}.
	\end{equation*}
\end{deff}

In the following theorem, we provide the OGFEA within our context.

\begin{theorem}[OGFEA]
	\label{th:OGFEA-M}
	For a tree $\mathcal{T}=(\mathcal{V},\mathcal{E})$, consider $\boldsymbol{N} = (N_v, \, v \in \mathcal{V})\sim\text{MPMRF}(\lambda,\boldsymbol{\alpha}, \mathcal{T})$, and suppose $M=\sum_{v\in\mathcal{V}}N_v$. The OGFEA of $N_v$ to $M$ is given by
	\begin{equation}
		\mathcal{P}^{[v]}_M (t) = \lambda \eta_{v}^{\mathcal{T}_v}(t\vecun{v})  \mathcal{P}_M(t), \quad t\in[-1,1]. 
		\label{eq:OGFEA}
	\end{equation} 
\end{theorem}
\begin{proof}
	From Theorem 2.4 of \cite{blier2022generating}, we have
	\begin{equation}
		\mathcal{P}_M^{[v]}(t) = \left. \left[ t_v  \frac{\partial}{\partial t_v} \mathcal{P}_{\boldsymbol{N}}(\boldsymbol{t}) \right] \right|_{\boldsymbol{t}=t\vecun{v}}.
		\label{eq:OGFEA-M-1}
	\end{equation}
	As we need to differentiate according to $t_v$, we prefer to root the underlying tree in $v$ since $\boldsymbol{t}_{u\mathrm{dsc}(u)}$, $u\in\mathcal{V}$, then contains $t_v$ only if $u=v$. The derivation would otherwise be complex due to (\ref{eq:jointpgf-h}) being recursively defined, and Theorem \ref{th:Reversibility} ensures the result is the same. Inserting (\ref{eq:JointPGF}) in (\ref{eq:OGFEA-M-1}), we obtain
	\begin{align}
		\mathcal{P}_M^{[v]}(t) &= \left. \left[ t_v  \frac{\partial}{\partial t_v} \prod_{i\in\mathcal{V}}\mathrm{e}^{\lambda(1-\alpha_{(\mathrm{pa}(i),i)})(\eta_{i}^{\mathcal{T}_v}(\boldsymbol{t}_{i\mathrm{dsc}(i)})-1)} \right] \right|_{\boldsymbol{t}=t\vecun{v}}\notag\\
		&=\left. \left[ t_v  \frac{\partial}{\partial t_v} \mathrm{e}^{\lambda(1-\alpha_{(\mathrm{pa}(v),v)})
			(\eta_{v}^{\mathcal{T}_v}(\boldsymbol{t}_{v\mathrm{dsc}(v)})-1)} \right] \right|_{\boldsymbol{t}=t\vecun{v}} \prod_{i\in\mathcal{V}\backslash\{v\}}\mathrm{e}^{\lambda(1-\alpha_{(\mathrm{pa}(i),i)})(\eta_{i}^{\mathcal{T}_v}(t\vecun{i})-1)} \notag\\
		&= \left. \lambda(1-\alpha_{(\mathrm{pa}(v),v)})\left[ t_v  \frac{\partial}{\partial t_v} 
		\eta_{v}^{\mathcal{T}_v}(\boldsymbol{t}_{v\mathrm{dsc}(v)}) \right] \right|_{\boldsymbol{t}=t\vecun{v}} \prod_{i\in\mathcal{V}}\mathrm{e}^{\lambda(1-\alpha_{(\mathrm{pa}(i),i)})(\eta_{i}^{\mathcal{T}_v}(t\vecun{i})-1)}. \notag
	\end{align}
	From (\ref{eq:jointpgf-h}), we notice $t_v  \frac{\partial}{\partial t_v} \eta_{v}^{\mathcal{T}_v}(\boldsymbol{t}_{v\mathrm{dsc}(v)}) = \eta_{v}^{\mathcal{T}_v}(\boldsymbol{t}_{v\mathrm{dsc}(v)})$. Moreover, $\alpha_{(\mathrm{pa}(v),v)} = 0$ by convention as we chose $v$ as the root. The result follows.
\end{proof}

From the OGFEA given in Theorem~\ref{th:OGFEA-M}, one directly identifies an analytic expression for the expected allocation by recovering the coefficients of (\ref{eq:OGFEA}) expressed in a polynomial form. This is presented in the following corollary. 

\begin{cor}[Expected Allocations]
	\label{th:ExpectedAllocs}
	For a tree $\mathcal{T}=(\mathcal{V},\mathcal{E})$, consider $\boldsymbol{N} = (N_v, \, v \in \mathcal{V})\sim\text{MPMRF}(\lambda,\boldsymbol{\alpha}, \mathcal{T})$, and suppose $M=\sum_{v\in\mathcal{V}}N_v$. Define $H_v^{\mathcal{T}_v}$ as a random variable with probability generating function given by $\eta_{v}^{\mathcal{T}_v}(t\vecun{v})$. The expected allocation of $N_v$ to $M$ is
	\begin{equation}
		\mathrm{E}[N_v\mathbb{1}_{\{M=k\}}] = \lambda \sum_{j=0}^{k} p_{H_v^{\mathcal{T}_v}}(k-j)p_M(j), \quad k\in\mathbb{N}.
		\label{eq:ExpeAllocs}
	\end{equation} 
\end{cor}

\begin{proof}
	We reexpress $\eta_{v}^{\mathcal{T}_v}(t\vecun{v})$ and  $\mathcal{P}_M(t)$ from (\ref{eq:OGFEA}) as polynomials in $t$. Since they are probability generating functions, their respective coefficients are given by $p_{H_v^{\mathcal{T}_v}}$ and $p_M$. Hence, we find
	\begin{equation}
		\mathcal{P}_{M}^{[v]}(t) = \lambda\left(\sum_{k=0}^{\infty}p_{H_v^{\mathcal{T}_v}}(k)t^k\right) \left(\sum_{k=0}^{\infty}p_{M}(k)t^k\right)\notag\\
		= \lambda \sum_{k=0}^{\infty}\left(\sum_{j=0}^k  p_{H_v^{\mathcal{T}_v}}(k-j)p_M(j) \right)t^k,
		\label{eq:ExpeAlloc-1}
	\end{equation}
	performing the Cauchy product. We identify the coefficient of $t^k$, $k\in\mathbb{N}$, in (\ref{eq:ExpeAlloc-1}) to be $\lambda \sum_{j=0}^{k} p_{H_v^{\mathcal{T}_v}}(k-j)p_M(j)$. This is the expected allocation to a total outcome $M=k$, according to Definition~\ref{def:OGFEA}.
\end{proof}

Under the assumption $\alpha_e=0$ for all $e\in\mathcal{E}$, which, as established above, corresponds to the case where the components of $\boldsymbol{N}$ are independent, $H_v^{\mathcal{T}_v}$ is a degenerate random variable at 1, $v\in\mathcal{V}$. This is deduced from (\ref{eq:jointpgf-h}), as then $\eta_{v}^{T_v}(t\vecun{v}) = t$ for all $v \in\mathcal{V}$. Hence, since $M$ itself then follows a Poisson distribution of parameter $\lambda d$, equation (\ref{eq:ExpeAllocs}) becomes
\begin{equation*}
	\mathrm{E}[N_v\mathbb{1}_{\{M=k\}}] = \lambda \sum_{j=0}^{k} \mathbb{1}_{\{k-j = 1\}}p_M(j) = \lambda p_M(k-1) = \lambda \frac{k}{\lambda d}  p_M(k),\quad k\in\mathbb{N},
\end{equation*}
which indeed corresponds to the expected allocation to the sum of independent Poisson  random variables, as provided notably in Section 10.3 of \cite{marceau2013modelisation}. From Theorem~\ref{th:OGFEA-M} and Corollary~\ref{th:ExpectedAllocs}, we note that the difference in the expected allocations of two components of the same $\boldsymbol{N}$, say $N_v$ and $N_w$, $v,w\in\mathcal{V}$, is only due to $\eta_{v}^{\mathcal{T}_v}$ and $\eta_{w}^{\mathcal{T}_w}$. Indeed, $\lambda$, $\mathcal{P}_M$ and $p_M$, as in (\ref{eq:OGFEA}) and (\ref{eq:ExpeAllocs}) are the same when considering $v$ or $w$.

Forthcoming Algorithm \ref{algo:OGFEA} leans on Section 2.3 of \cite{blier2022generating} to provide a method for the computation of expected allocations of $N_v$ to $M$ using the FFT algorithm. It relies on Theorem \ref{th:OGFEA-M} and, as in Algorithm \ref{algo:fftM}, exploits the recursiveness of (\ref{eq:jointpgf-h}) in Theorem~\ref{th:JointPGF}. A combined use of Algorithms \ref{algo:fftM} and \ref{algo:OGFEA} allows to compute conditional mean risk-sharing. Algorithm~\ref{algo:OGFEA} is relevent although Corollary~\ref{th:ExpectedAllocs} provides an analytic expression for expected allocations, as it is more efficient computationwise than the programming of the convolution product in (\ref{eq:ExpeAllocs}) if one were to compute expected allocations for multiple values of $k$, $k\in\mathbb{N}$.

\begin{algorithm}[H]
	\label{algo:OGFEA}
	\caption{Computing expected allocations of $N_v$ to $M$.} 
	\KwIn{Weighted adjacency matrix $\boldsymbol{A} = (A_{ij})_{i\times j\in \mathcal{V}\times \mathcal{V}}$; $\lambda$.}
	\KwOut{Vector $\boldsymbol{p}^{(M)}=(p^{(M)}_k)_{k\in\{1,\ldots,k_{\max}\}} $ such that $p^{(M)}_k$ = $\mathrm{E}[N_v\mathbb{1}_{\{M=k-1\}}]$.}
	 Modify $\boldsymbol{A}$ to be topologically ordered accordingly to root $v$ using Algorithm \ref{algo:changematrix}, \ref{sect:AdditionalAlgo}.\\
	 Follow all the steps of Algorithm \ref{algo:fftM}, but replace (\ref{eq:algo}) by\\ 
	 $$ \phi_l^{(M)} = \lambda h_1 \prod_{k} \exp(\lambda(1-A_{\pi_k,k})(h_k-1)).$$
\end{algorithm}

\section{Numerical example}
\label{sect:NumericalExamples}

We present an example to synthesize some of the findings from the previous sections. In particular, the aim is to illustrate the impact of altering the dependence parameters on the distribution of $M$ and the associated contributions to the TVaR under Euler's principle $\mathcal{C}_{\kappa}^{\mathrm{TVaR}}(N_v;M)$, $v\in\mathcal{V}$. Consider $\boldsymbol{N}\sim$ MPMRF$(\lambda,\boldsymbol{\alpha}, \mathcal{T})$, with $\lambda = 1$ and the 50-vertex tree $\mathcal{T}$ depicted in Fig.~\ref{fig:Pando}. Consider four vectors of dependence parameters: $\boldsymbol{\alpha}^{(0)}$, $\boldsymbol{\alpha}^{(0.3)}$, $\boldsymbol{\alpha}^{(0.7)}$ and $\boldsymbol{\alpha}^{(0.9)}$, such that $\boldsymbol{\alpha}^{(\beta)} = (\beta,\ldots, \beta) $, $\beta\in[0,1]$, that is, for each vector, we suppose identical dependence parameter for all edges. Assume $\boldsymbol{N}^{(\beta)}\sim$ MPMRF$(\lambda,\boldsymbol{\alpha}^{(\beta)},\mathcal{T})$, with $\lambda = 1$ and the tree $\mathcal{T}$ depicted in Fig.~\ref{fig:Pando}, and correspondingly $M^{(\beta)} = \sum_{v\in\mathcal{V}}N_v^{(\beta)}$, $\beta\in[0,1]$. We are interested in the distribution of $M^{(\beta)}$. 

\begin{minipage}[b]{0.45\textwidth}
               \begin{figure}[H] 
		\centering
		\begin{tikzpicture}[every node/.style={text=Black, circle, draw = Maroon, inner sep=0mm, outer sep = 0mm, minimum size=3mm, fill = White}, node distance = 1.5mm, scale=0.1, thick]
		\node (1a) {\tiny $1$};
		\node [right=of 1a] (2a) {\tiny $2$};
		\node [above left = of 2a] (3a) {\tiny $3$};
		\node [above = of 2a] (4a) {\tiny $4$};
		\node [above right =of 2a] (5a) {\tiny $5$};
		\node [below left=of 2a] (6a) {\tiny $6$};
		\node [below  =of 2a] (7a) {\tiny $7$};
            \node [below right =of 2a] (8a) {\tiny $8$};
            \node [right  = 0.875cm of 2a] (9a) {\tiny $9$};
            \node [above left  = of 9a] (10a) {\tiny $10$};
            \node [above = of 9a] (11a) {\tiny $11$};
            \node [above right = of 9a] (12a) {\tiny $12$};
            \node [below left  = of 9a] (13a) {\tiny $13$};
            \node [below  = of 9a] (14a) {\tiny $14$};
            \node [below right = of 9a] (15a) {\tiny $15$};
            \node [right  = 0.875cm of 9a] (16a) {\tiny $16$};
            \node [above left  = of 16a] (17a) {\tiny $17$};
            \node [above  = of 16a] (18a) {\tiny $18$};
            \node [above right  = of 16a] (19a) {\tiny $19$};
            \node [below left  = of 16a] (20a) {\tiny $20$};
            \node [below   = of 16a] (21a) {\tiny $21$};
            \node [below right  = of 16a] (22a) {\tiny $22$};
            \node [right  = 0.875cm of 16a] (23a) {\tiny $23$};
            \node [above left  = of 23a] (24a) {\tiny $24$};
            \node [above  = of 23a] (25a) {\tiny $25$};
            \node [above right  = of 23a] (26a) {\tiny $26$};
            \node [below left  = of 23a] (27a) {\tiny $27$};
            \node [below   = of 23a] (28a) {\tiny $28$};
            \node [below right  = of 23a] (29a) {\tiny $29$};
            \node [right  = 1.675cm of 23a] (30a) {\tiny $30$};
            \tikzset{shift=(30a)};
            \foreach \name in {31,...,50}
				{
             \node (\name) at ({(-\name*17.1-386.5)}: 11.5cm) {\tiny $\name$};
             }
             \node [draw = none, below  = 0.25cm of 46] (phantom) {};
		
		\draw (1a) -- (2a);
		\draw (2a) -- (3a);
		\draw (2a) -- (4a);
		\draw (2a) -- (5a);            
		\draw (2a) -- (6a);       
		\draw (2a) -- (7a);
            \draw (2a) -- (8a);
            \draw (2a) -- (9a);
            \draw (9a) -- (10a);
            \draw (9a) -- (11a);
            \draw (9a) -- (12a);
            \draw (9a) -- (13a);
            \draw (9a) -- (14a);
            \draw (9a) -- (15a);
            \draw (9a) -- (16a);
            \draw (16a) -- (17a);
            \draw (16a) -- (18a);
            \draw (16a) -- (19a);
            \draw (16a) -- (20a);
            \draw (16a) -- (21a);
            \draw (16a) -- (22a);
            \draw (16a) -- (23a);
            \draw (23a) -- (24a);
            \draw (23a) -- (25a);
            \draw (23a) -- (26a);
            \draw (23a) -- (27a);
            \draw (23a) -- (28a);
            \draw (23a) -- (29a);
            \draw (23a) -- (30a);
            \draw (30a) -- (31);
            \draw (30a) -- (32);
            \draw (30a) -- (33);
            \draw (30a) -- (34);
            \draw (30a) -- (35);
            \draw (30a) -- (36);
            \draw (30a) -- (37);
            \draw (30a) -- (38);
            \draw (30a) -- (39);
            \draw (30a) -- (40);
            \draw (30a) -- (41);
            \draw (30a) -- (42);
            \draw (30a) -- (43);
            \draw (30a) -- (44);
            \draw (30a) -- (45);
            \draw (30a) -- (46);
            \draw (30a) -- (47);
            \draw (30a) -- (48);
            \draw (30a) -- (49);
            \draw (30a) -- (50);
		\end{tikzpicture}
		\caption{Tree $\mathcal{T}$ for the numerical example.}
		\label{fig:Pando}
	\end{figure}
\end{minipage}
\hfill
\begin{minipage}[b]{0.5\textwidth}
 \begin{figure}[H]
\centering
\includegraphics[width=\textwidth]{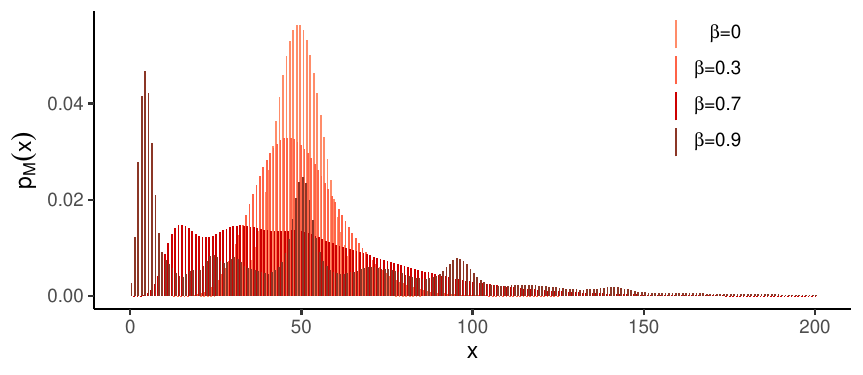}
\caption{Probability mass functions of $M^{(\beta)}$, $\beta\in\{0,0.3,0.7,0.9\}$.}
\label{fig:pmfM}
\end{figure}
\end{minipage}

 Referring to Algorithm~\ref{algo:fftM}, we compute the probability mass functions of $M^{(\beta)}$, $\beta\in\{0,0.3,0.7,0.9\}$; their depictions are provided in Fig.~\ref{fig:pmfM}.
    The case $\beta=0$ corresponds to the limit distribution where all components of $\boldsymbol{N}$ are independent and thus serves as a comparison standard in assessing the impact of the dependence in the cases $\beta = 0.3,0.7,0.9$. For $\beta = 0$, $M$ is expected to behave as a Gaussian distribution, given the central limit theorem; indeed, the shape of the probability mass function of $M^{(0)}$ seems to confirm this. One may suppose this behavior would also surface if the dependence between components is weak overall. While correlation coefficients of 0.3 and higher may not necessarily be considered weak dependence, one must not forget that the components of the vector $\boldsymbol{\alpha}$ multiply themselves following the path between vertices -- recall (\ref{eq:rhoNj1Nj2}) -- so that, in our context, $\beta$ acts as an exponential base to compute correlations. Hence, given the tree $\mathcal{T}$, for every pair of distinct vertices $u,v\in\mathcal{V}$, one has $\beta^{6}\leq\rho_{P}(N_u^{(\beta)},N_v^{(\beta)})\leq\beta$, with $\beta\in[0,1]$. A 70\% reduction in dependence per travelled edge, given by $\beta=0.3$, reduces dependence enough to produce a
    Gaussian silhouette for the probability mass function of $M^{(0.3)}$ in Fig.~\ref{fig:pmfM}. The silhouette nonetheless skews a little, assigning more mass to larger values than expected from a normal distribution. As $\beta$ becomes larger, take $\beta=0.7$ for example, this skewness is amplified. Also,
    the propagation part in $(\ref{eq:StoDynamics})$ becomes more potent than the innovation part,
    to the extent that, for $\beta=0.7$, events can often propagate to multiple vertices, resulting in a multimodal probability mass function.  Multimodality is then intensified in the case $\beta=0.9$, where clumps of masses at multiples of 50 clearly show that most components of $\boldsymbol{N}$ take the same value -- there is little place for the innovation part. The anfractuosities around those clumps
    are exhibits of the peculiar shape of $\mathcal{T}$. To support this reasoning, we provide in Fig.~\ref{fig:pmfCM} the depiction of the probability mass function of $C_{M^{(\beta)}}$, the secondary random variable under the principle that $M^{(\beta)}$ follows a compound Poisson distribution (Theorem~\ref{th:DistrM}), for $\beta\in\{0,0.3,0.7,0.9\}$. 
    \begin{figure}[H]
\centering
\begin{tikzpicture}
    \node[anchor = south west] at (0,0) {\includegraphics[width = 0.3\textwidth]{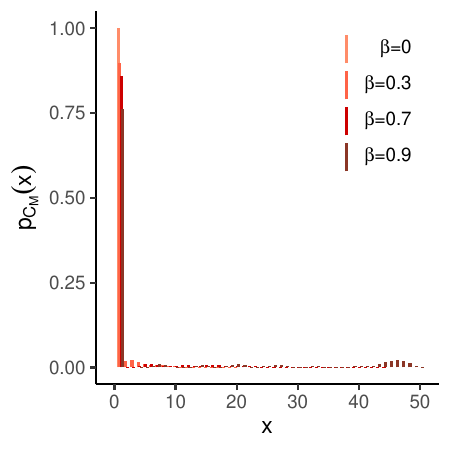}};
    \node[anchor = south west] at (5.75,0.5) {\includegraphics[width = 0.4\textwidth]{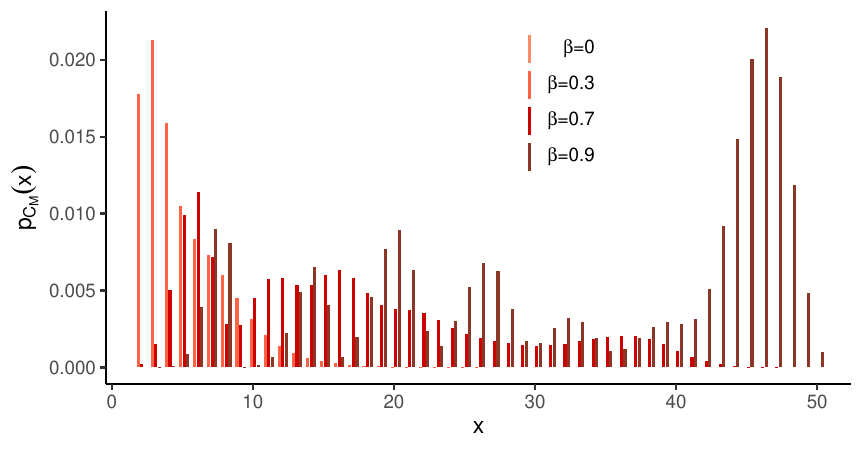}};

    \draw[Black!50, thick] (1.51,0.96) rectangle (5,1.2); 
    	\draw[Black!50, very thick, ->] (4.4,0.96) to [out=315, in=225] (7,1);
     \filldraw[draw = white, fill = white] (9.7,2.7) rectangle (11,4);
\end{tikzpicture}
\caption{Probability mass functions of $C_{M^{(\beta)}}$, $\beta\in\{0,0.3,0.7,0.9\}$.}
\label{fig:pmfCM}
\end{figure}
   Note how mass transports itself from one extreme, $x=1$, to another, $x=50$, as the dependence parameter grows. Obviously, $C_{M^{(0)}}$ is degenerate at 1 since $M^{(0)}$ follows a simple Poisson distribution by the independence of the components of $\boldsymbol{N}^{(0)}$. The random variable $C_{M^{(\beta)}}$ represents the number of vertices affected per batch of events, while $\lambda_M$ governs the number of such batches. The prominence of the propagation part thus indeed manifests itself through $C_{M^{(\beta)}}$. Table~\ref{tab:compoundM} illustrates how $\lambda_{M^{(\beta)}}$ decreases as $\mathrm{E}[C_{M^{(\beta)}}]$ increases -- batches occur less, but comprise more events --, one balancing out the other so that $\mathrm{E}[M^{(\beta)}]=50$ for every $\beta\in\{0,0.3,0.7,0.9\}$. 
\begin{table}[H]
\centering
\begin{tabular}{lrr}
\hline
&$\lambda_{M^{(\beta)}}$ &$\mathrm{E}[C_{M^{(\beta)}}]$\\
\hline
 $\beta=0$ & 50 & 1\\
 $\beta=0.3$ & 35.3& 1.416\\
 $\beta=0.7$ & 15.7& 3.185\\
 $\beta=0.9$ & 5.9& 8.475\\
\hline	
\end{tabular}
\caption{Primary and secondary distributions' expectations $\lambda_{M^{(\beta)}}$ and $\mathrm{E}[C_{M^{(\beta)}}]$, $\beta\in\{0,0.3,0.7,0.9\}$.} 
\label{tab:compoundM}
\end{table}

    According to Theorem~\ref{th:ConvexOrderM}, $\boldsymbol{N}^{(\beta)}$'s are supermodularly ordered following the values of $\beta$. The consequential convex ordering 
    \begin{equation}
M^{(0)}\preceq_{cx}M^{(0.3)}\preceq_{cx}M^{(0.7)}\preceq_{cx}M^{(0.9)},
        \label{eq:Mordering}
    \end{equation}
provided by Corollary~\ref{th:ConvexOrderM}, may be observed from the previously computed probability mass functions of $M$. In this vein, we depict in Fig.~\ref{fig:cdf-slM} their corresponding cumulative distribution functions and stop-loss functions. The latter, for a random variable $X$, is given by ${\pi}_X(x)=\mathrm{E}[\max(X-x,0)]$, $x\in\mathbb{R}_+$. 

	\begin{figure}[H]
	\centering
	\begin{subfigure}{0.45\textwidth}
	\centering
	\label{subfig:cdfMalphR}
	\includegraphics[width=\textwidth]{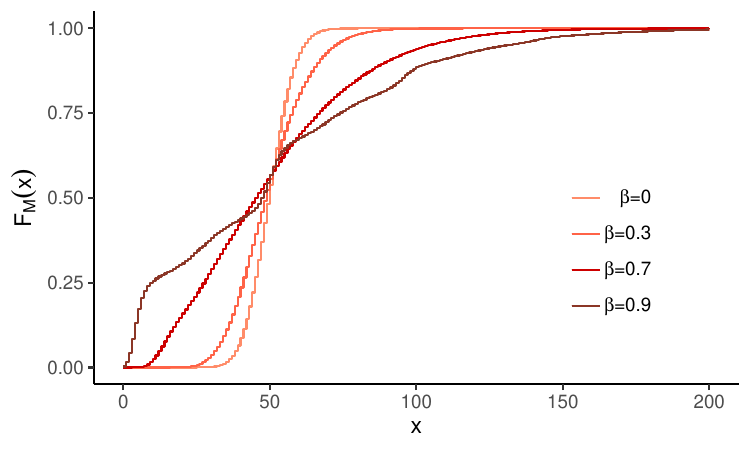} 
\end{subfigure}
\begin{subfigure}{0.45\textwidth}
\centering
\label{subfig:slMalphR}
\includegraphics[width=\textwidth]{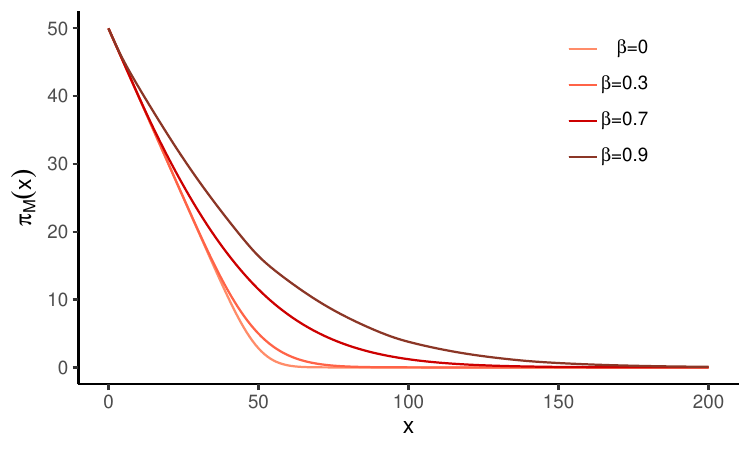} 
\end{subfigure}
\caption{Cumulative distribution functions and stop-loss functions of $M^{(\beta)}$, $\beta\in\{0,0.3,0.7,0.9\}$.}
\label{fig:cdf-slM}
\end{figure}

For $\beta=0,0.3,0.7$, the cumulative distribution functions' curves cross each other at a single point and, hence, directly confirm the ordering 
in (\ref{eq:Mordering}) by the Karlin-Novikoff sufficient criterion for convex ordering. However, the squiggly shape of the curve when $\beta=0.9$ prevents from applying it for the comparison with the cumulative distribution function when $\beta=0.7$. Nonetheless, the stacking of the stop-loss curves in the right-hand side plot shows that the ordering in (\ref{eq:Mordering}) indeed holds. As a result, variances and convex risk measures, such as the TVaR and the entropic risk measure, given by $\varphi^{\mathrm{ent}}_{\rho}(X) = \tfrac{1}{\rho}\mathrm{ln}\mathrm{E}[\mathrm{e}^{\rho X}]$, $\rho\in(0,\sup\{\rho_0>0: \mathrm{E}[\mathrm{e}^{\rho_0 X}]<\infty\})$, are ordered accordingly. This is exhibited in Table~\ref{tab:M-Riskmeasures}.
While VaR is not a convex risk measure, high-quantile VaRs also agree with the order of riskiness given in (\ref{eq:Mordering}). 

\begin{table}[H]
\centering
\begin{tabular}{lrrrrrr}
\hline
&$\mathrm{Var}(M^{(\beta)})$&$\mathrm{VaR}_{0.9}(M^{(\beta)})$&$\mathrm{TVaR}_{0.9}(M^{(\beta)})$&$\mathrm{VaR}_{0.99}(M^{(\beta)})$&$\mathrm{TVaR}_{0.99}(M^{(\beta)})$&$\varphi^{\mathrm{ent}}_{0.1}(M^{(\beta)})$\\
\hline
$\beta=0$ & 50.00 & 59 &62.76&67&69.82&52.59\\
$\beta=0.3$ & 157.84 & 67 & 74.60& 84& 90.85& 60.47\\
$\beta=0.7$ & 841.80 & 90& 109.95&  135& 152.30& 200.01\\
$\beta=0.9$ & 1762.60 & 105& 137.35& 175& 199.38&719.77\\
\hline	
\end{tabular}
\caption{Variances and risk measures of $M^{(\beta)}$, $\beta\in\{0,0.3,0.7,0.9\}$.} 
\label{tab:M-Riskmeasures}
\end{table}
 
We examine contributions to the TVaR under Euler's principle, defined in equation (\ref{eq:contribTVaR}). For brevity, we focus on the contributions of random variables associated to three vertices, namely $v\in\{1,16,30\}$. The values are provided in Table~\ref{tab:contribTVaR}, with the fraction of the corresponding $\mathrm{TVaR}(M^{(\beta)})$ placed in parentheses.   
Values in Table~\ref{tab:contribTVaR}, as well as all other values computed in the context of the present numerical example, were obtained through exact-computation methods, implemented in Algorithms~\ref{algo:fftM} and~\ref{algo:OGFEA}. We did not resort to any simulation.

\begin{table}[H]
\centering
\begin{tabular}{lrrr}
\hline
&$\mathcal{C}^{\mathrm{TVaR}}_{0.9}(N_{1}^{(\beta)}, M^{(\beta)})$ & $\mathcal{C}^{\mathrm{TVaR}}_{0.9}(N_{16}^{(\beta)}, M^{(\beta)})$ &$\mathcal{C}^{\mathrm{TVaR}}_{0.9}(N_{30}^{(\beta)}, M^{(\beta)})$\\
\hline
 $\beta=0$ &  1.26 (2.00\%)&  1.26 (2.00\%)&  1.26 (2.00\%)\\
 $\beta=0.3$ &  1.31 (1.75\%)& 1.86 (2.50\%)&  2.38 (3.19\%)\\
 $\beta=0.7$ &  1.88 (1.71\%)&  2.63 (2.40\%)&  2.75 (2.50\%)\\
 $\beta=0.9$ &  2.58 (1.88\%)&  2.95 (2.14\%)&  2.96 (2.15\%)\\
\hline	
\end{tabular}
\caption{Contributions to the TVaR under Euler' principle of $N_v^{(\beta)}$ to $M^{(\beta)}$, $v\in\{1,16,30\}$, $\beta\in\{0,0.3,0.7,0.9\}$.} 
\label{tab:contribTVaR}
\end{table}
For each $v\in\{1,16,30\}$, the contribution $\mathcal{C}^{\mathrm{TVaR}}_{0.9}(N_v^{(\beta)},M^{(\beta)})$ monotonically increases along the value of $\beta$; this is intuitive since there is indeed more to allocate, cf. values of $\mathrm{TVaR}_{0.9}(M^{(\beta)})$ in Table~\ref{tab:M-Riskmeasures}. 
As expected, every component of $\boldsymbol{N}^{(0)}$ contributes equally 2\%, a fiftieth, of $\mathrm{TVaR}(M^{(0)})$. Otherwise, for $\beta\in\{0.3,0.7,0.9\}$, contributions differ according to the position of the vertices on the tree, as supported by Corollary~\ref{th:ExpectedAllocs}.  We observe that the gap between the contributions grows, then shrinks, when $\beta$ increases, seemingly converging to a point where every vertex contributes, again, a fiftieth of the TVaR. Indeed, when dependence is very strong, all components of $\boldsymbol{N}$ tend to take the same value, and hence, allocation methods should not distinguish them. Furthermore, vertex~30, in the middle of a group of numerous vertices, seems to yield larger contributions than vertices~1 and 16. Although stochastic orderings comparing the strength of the dependence of a component of $\boldsymbol{N}$ to the aggregate random variable $M$ -- for instance, $(N_{1}^{(\beta)},M^{(\beta)}) \preceq_{sm} (N_{30}^{(\beta)},M^{(\beta)})$ for $\beta\in [0,1]$ -- would have to be established to prove that this ordering of contributions effectively holds.

Let us revisit our example with values of every component of the vector of $\boldsymbol{\alpha}$ as indicated in Table~\ref{tab:alphadistinct}; they are no longer identical.
Within every star-like substructure constituting $\mathcal{T}$ -- four small and one bigger -- random variables are strongly correlated, and for the edges connecting these substructures the dependence parameter is more moderate. Also, assume $\lambda = 0.1$, making events' occurrences rarer. In Fig.~\ref{fig:pmfMalphadiff}, we depict the probability mass function of $M$, whose values were again computed using Algorithm~\ref{algo:fftM}. 

\begin{minipage}[b]{0.45\textwidth}
\begin{table}[H]
    \centering
    \begin{tabular}{lr}
    \hline
        $v$ & $\alpha_{(\mathrm{pa}(v),v)}$ \\
    \hline
         $\{2,3,4,5,6,7,8\}$& 0.8\\
         $\{10,11,12,13,14,15\}$& 0.6\\
         $\{17,18,19,20,21,22\}$& 0.5\\
         $\{24,25,26,27,28,29\}$& 0.4\\
         $\{31,32,\ldots,49,50\}$& 0.9\\
         $\{9,16,23,30\}$& 0.1\\
    \hline
    \end{tabular}
    \caption{Dependence parameter between every vertex and its parent, assuming the root is 1 for notational convenience.}
    \label{tab:alphadistinct}
\end{table}
\end{minipage}
\hfill
\begin{minipage}[b]{0.5\textwidth}
\begin{figure}[H]
    \centering
    \includegraphics[width = 0.8\textwidth]{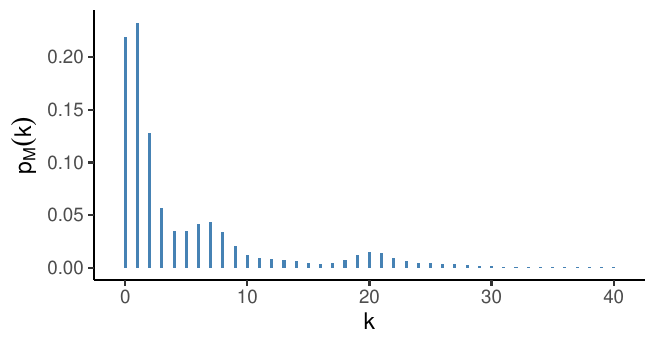}
    \caption{Probability mass function of $M$ given the vector of dependence parameters $\boldsymbol{\alpha}$ given in Table~\ref{tab:alphadistinct}.}
    \label{fig:pmfMalphadiff}
\end{figure}
\end{minipage}

Previously, the multimodality of the probability mass function followed from the similarity in all components of $\boldsymbol{N}$, clumps corresponding roughly to $N_v=0,1,2,\ldots$, for most $v\in\mathcal{V}$. In this case, the lower value of $\lambda$ makes the probability of $N_v\geq 2$ negligible; bumps of probability mass are somewhat due to the choice of dependence parameter combined with the tree's shape. Events occur rarely, but as they do, they happen in conjunction with others at vertices within the same star-like substructure. The 21 vertices of the bigger substructure and its edges' dependence parameters' value of 0.9, create the bump at $k=20$, such that $\mathrm{Pr}(M> 20) = 0.0626$ is not negligible.

\begin{figure}[t]
	\centering
	\begin{subfigure}{0.45\textwidth}
	\centering
	\label{subfig:alphdiffR}
    \begin{tikzpicture}
    \node[anchor = south west, outer sep = 0, inner sep = 0] at (-1,-1) {\includegraphics[width = \textwidth]{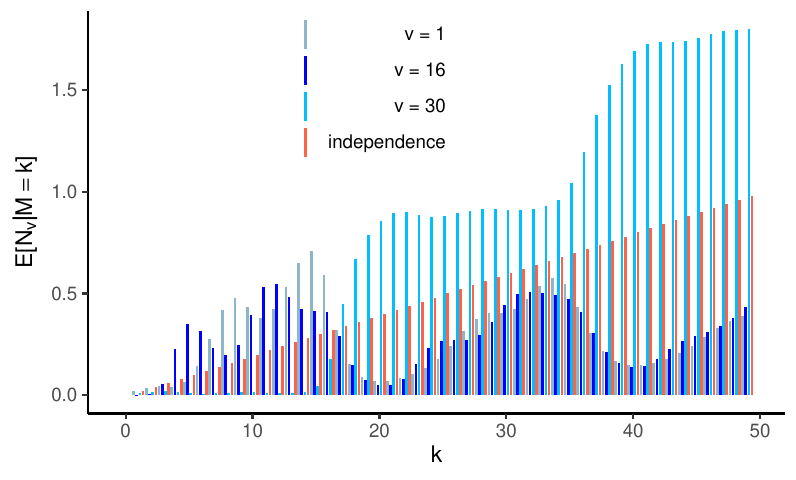}};
\end{tikzpicture}
\end{subfigure}
\begin{subfigure}{0.45\textwidth}
\centering
\begin{tikzpicture}
    \node[anchor = south west, outer sep = 0, inner sep = -1, rectangle] at (0,0) {\includegraphics[width = \textwidth]{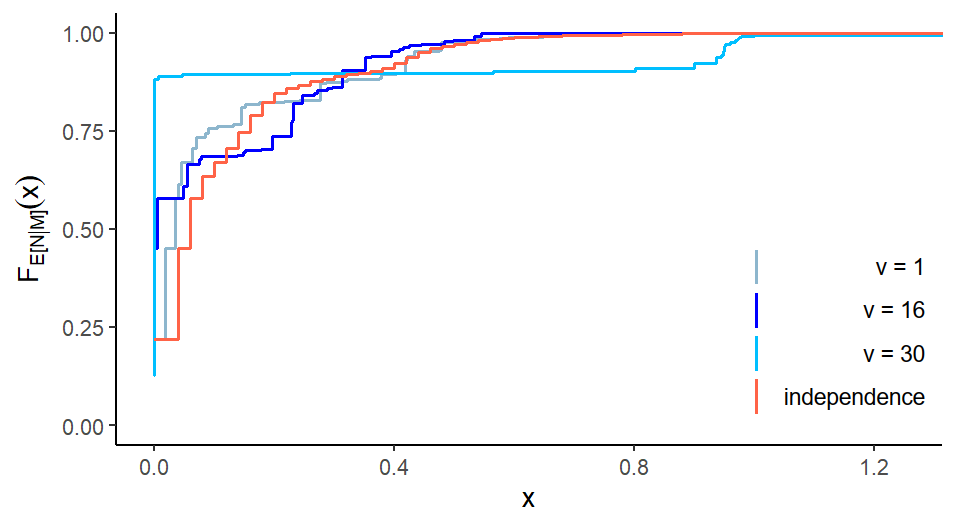}};
    \draw[Black!50, thick]   (1.6,3.7) -- (1.6,0.2) node[text = Black!50, outer sep = 0, inner sep = 0, below = 2pt]{\scriptsize $\mathrm{E}[N_v]$}; 
\end{tikzpicture}
\label{subfig:emptyR}
\end{subfigure}
\caption{(Left) Conditional mean risk sharing of $N_v$ to $M$, $v\in\{1,16,30\}$. (Right) Cumulative distribution function of the random variable $\mathrm{E}[N_v|M]$, $v\in \{1,16,30\}$.}
\label{fig:cmrs}
\end{figure}

We compute values of conditional mean risk sharing of $N_v$ to $M$, as described in (\ref{eq:cmrs}), for $v\in\{1,16,30\}$ by employing Algorithm~\ref{algo:OGFEA} and the previously computed probability mass function. 
Fig.~\ref{fig:cmrs} illustrates these values for $k\in\{0,1,\ldots,50\}$. 
In the case of a vertex within the most extensive star-like substructure, $v=30$ for our consideration, the part allocated by conditional mean risk sharing
undergoes significant increases at $k$ around multiples of 21. 
This reflects the effect of the more substantial dependence parameter for the edges within this substructure than in the rest of the tree. Moreover, because of the prominence of the substructure within the tree, this behavior reverberates into the conditional mean risk sharing of the other vertices as it decreases upon approaching these multiples. Such behavior is induced by the specific dependence scheme and choice of a small $\lambda$; this comes in contrast with often observed monotone increasingness in $k$ (see Proposition 4.3 of \cite{denuit2024conditional}) as would be the case for independent $N_v$, $v\in\mathcal{V}$. We furthermore depict in Figure~\ref{fig:cmrs} the cumulative distribution function of the random variable $\mathrm{E}[N_v|M]$, $v\in\{1,16,30\}$, providing an understanding of the allocation's dynamics from every component of $\boldsymbol{N}$'s individual point of view. 
While every of these random variables shares the same expectation $\mathrm{E}[\mathrm{E}[N_v|M]] = 0.1$, unaffected by the dependence scheme, their distribution 
depends both on the strength of the dependence between $N_v$ and the other components of $\boldsymbol{N}$ and the strength of the dependence between these other components. For instance, the cumulative distribution function of $\mathrm{E}[N_{16}|M]$ dominates that of $\mathrm{E}[N_v|M]$ under independence for values of $x>0.36$: cases in which a larger $M$ has to be allocated see vertices in the bigger star-like substructure absorb more of the allocation, such that $\mathrm{E}[N_{16}|M]$ rarely takes larger values. 
Such behavior is influenced by the vertex's relative position within the tree and the chosen dependence parameters $\boldsymbol{\alpha}$ and mean parameter $\lambda$.

\section{Conclusion}
\label{sect:Conclusion}
We have introduced a new family of tree-structured MRFs with fixed Poisson marginal distributions. For previous families of MRFs, a burden comes from their constructions' resting on conditional distributions, rendering recondite marginal and joint distributions. The construction proposed in Theorem~\ref{th:StoDynamics}, on the contrary, allows to derive analytic expressions for its joint probability mass function and joint probability generating function. Studying covariances, we have observed that positive dependence
ties components of MRFs from the proposed family. We have been able to establish stochastic orderings between the MRFs. We studied the distribution of the sum of the MRF's components, derived its OGFEA, and established stochastic orderings for it as well.

 Finally, while addressed in Sections~\ref{subsect:StochasticOrderingN} and~\ref{sect:Sum}, the effect of the tree's topology on the dependence relations deserves to be investigated further: it is, after all, the main feature of a tree-structured MRF. Remark~\ref{th:NoSupermodularNShape} indicates that no supermodular ordering may be established between members of $\mathbb{MPMRF}$ defined on trees of different topologies, the sum of the components may nonetheless be comparable according to the convex order. This invites questions such as : "Which of two different underlying trees yields a riskier distribution of $M$?" or "Which of two vertices hosts the random variable having the most influence on the MRF's aggregate dynamics?". 
We believe such questions, inviting the notions of ordering of trees and vertices' centrality, have not yet been addressed for any family of MRF. Thus opening what we believe are new grounds of analysis, we preferred to devote a full paper on the matter. In \cite{cote2024mrfshape}, we show that, in fact, these two questions are both sides of the same coin and provide an answer through the designing of a new partial order of tree topologies.

\section*{Acknowledgment}
This work was partially supported by the Natural Sciences and Engineering Research Council of Canada (Cossette: 04273; Marceau: 05605; Côté), by the Fonds de recherche du Québec -- Nature et technologie (Côté: 335878), and by the Chaire en actuariat de l'Université Laval (Cossette, Côté, Marceau). We would like to thank the editor and anonymous referees for the comments that helped to improve the quality of this paper.
				

				
				
				\appendix

				\section{The binomial thinning operator}
				\label{sect:BinomThinnOp}
				
				\begin{theorem}[Properties of binomial thinning] The binomial thinning operator, denoted by $\circ$, and whose definition is given in (\ref{eq:BinomThinning}), has the following properties:
					\begin{enumerate} [nosep]
						\item[\rm(a)] Expectation: $\mathrm{E}[\alpha\circ X] = \alpha \mathrm{E}[X]$;
						\item[\rm(b)] Variance: $\mathrm{Var}(\alpha \circ X) = \alpha^2\mathrm{Var}(X) + \alpha(1-\alpha)\mathrm{E}[X]$;
						\item[\rm(c)] Covariance: $\mathrm{Cov}( X,\alpha\circ X) = \alpha\mathrm{Var}(X)$;
						\item[\rm(d)] 	\label{th:PropertyBinThinOp-PGF}
						Probability generating function: $\mathcal{P}_{\alpha\circ X}(t) = \mathcal{P}_X(1-\alpha + \alpha t)$;
						\item[\rm(e)] Distributivity in distribution: $\alpha\circ(X_1+X_2) \stackrel{d}{=} \alpha \circ X_1 + \alpha \circ X_2$;
						\item[\rm(f)] Associativity in distribution: $\alpha_1\circ(\alpha_2\circ X)\stackrel{d}{=}(\alpha_1\alpha_2)\circ X \stackrel{d}{=} \alpha_2\circ(\alpha_1\circ X)$;
						\item[\rm(g)] Stochastic dominance: $(\alpha \circ X | X = \theta_1) \preceq_{st} (\alpha \circ X | X = \theta_2)$, for $\theta_1, \theta_2 \in \mathbb{N}$, $\theta_1 \leq \theta_2$;
					\end{enumerate}
					with $X,X_1,X_2$ as random variables and $\alpha,\alpha_1,\alpha_2 \in [0,1]$.
					\label{th:PropertyBinThinOp}
				\end{theorem}
				\begin{proof}
					Properties (a) through (d) are easily derived from the construction in (\ref{eq:BinomThinning}). Using Property (d), one can verify Property (e) as follows: 
					\begin{align*}
				\mathcal{P}_{\alpha\circ(X_1+X_2)}(t) = \mathcal{P}_{X_1+X_2}(1-\alpha + \alpha t) 
                        = \mathcal{P}_{X_1,X_2}(1-\alpha + \alpha t, 1-\alpha + \alpha t)
                        = \mathcal{P}_{\alpha\circ X_1, \alpha\circ X_2}(t,t)=  \mathcal{P}_{\alpha \circ X_1 + \alpha \circ X_2}(t).
					\end{align*}
					For Property (f), we have 
					\begin{align*}
				\mathcal{P}_{\alpha_1\circ(\alpha_2\circ X)}(t) = \mathcal{P}_{\alpha_2\circ X}(1-\alpha_1 + \alpha_1 t)
						= \mathcal{P}_X(1-\alpha_1 + \alpha_1(1-\alpha_2 + \alpha_2 t))
						 = \mathcal{P}_X(1-\alpha_1\alpha_2 + \alpha_1\alpha_2 t)
						= \mathcal{P}_{(\alpha_1\alpha_2)\circ X}(t). 
					\end{align*}
					Property (g) follows from the fact that $(\alpha \circ X | X = \theta)$ is Binomial distributed with number of trials $\theta$ and success probability $\alpha$, $\theta \in \mathbb{N}$, $\alpha \in [0,1]$. 
				\end{proof}
				
				\section{Additional algorithm}
				\label{sect:AdditionalAlgo}
				
				We provide an algorithm to change the root in accordance to which a weighted adjacency matrix is constructed. 
				
				\begin{algorithm}[H]
					\label{algo:changematrix}
					\caption{Changing root in accordance to with an adjacency matrix is in topological order.} 
					\KwIn{Weighted adjacency matrix $\boldsymbol{A} = (A_{ij})_{i\times j\in \mathcal{V}\times \mathcal{V}}$ constructed in topological order according to $r$.}
					\KwOut{Weighted adjacency matrix $\boldsymbol{A}^{\prime} = (A^{\prime}_{ij})_{i\times j\in \mathcal{V}\times \mathcal{V}}$ constructed in topological order according to $r^{\prime}$.}
					 \For{$k = d,d-1,\ldots,2$} {
						 Compute $\pi_k = \inf\{j:A_{kj} >0\}$.\\
					}
					 Set $w = r^{\prime}$.\\
					 \While{$w\neq r$}{
						 Add $w$ to a vector $\mathrm{\mathbf{w}}$.\\
						 Overwrite $w$ by $\pi_{w}$.\\
					}
					 \For{$i \in \mathrm{\mathbf{w}}$}{
						 \For{$j\in i-1,i-2,\ldots,1$}{
							 Switch rows $j$ and $j+1$ of $\boldsymbol{A}$.\\
							 Switch columns $j$ and $j+1$ of $\boldsymbol{A}$.\\}}
					 Return $\boldsymbol{A}$.\\
				\end{algorithm}

\bibliographystyle{apalike} 
\bibliography{bibPoissonAR1}

					\end{document}